\documentclass[conference,final]{IEEEtran}
\pdfoutput=1
\usepackage{cite}
\usepackage{amsmath,amssymb,amsfonts}

\def\BibTeX{{\rm B\kern-.05em{\sc i\kern-.025em b}\kern-.08em
    T\kern-.1667em\lower.7ex\hbox{E}\kern-.125emX}}

\usepackage{tikz}
\usepackage{bussproofs}
\usepackage{amsthm}
\usepackage{mathtools}
\usepackage{algorithm}
\usepackage[noend]{algpseudocode}
\usepackage{subcaption}
\usepackage{wrapfig}
\usepackage{hyperref}
\usepackage{cleveref}
\usepackage{pifont}
\usepackage{thm-restate}
\usepackage{booktabs}
\usepackage{multirow}
\usepackage{colortbl}
\usepackage{accents}
\usepackage{ltl}
\usepackage{balance}

\def\proofsnamefont{\itshape}
\def\proofsindent{\noindent}
\newenvironment{proofSketch}[1][Proof sketch]{\par
    \pushQED{\qed}%
    \normalfont 
    \trivlist
    \item[\proofsindent\hskip\labelsep
    {\proofsnamefont #1.}]\ignorespaces
}{%
    \popQED\endtrivlist
}

\newtheorem{theorem}{Theorem}
\newtheorem{remark}{Remark}
\newtheorem{corollary}{Corollary}
\newtheorem{lemma}{Lemma}
\newtheorem{proposition}{Proposition}
\newtheorem{definition}{Definition}
\newtheorem{example}{Example}

\newtheorem{assumption}{Assumption}

\newcommand{\cmark}{\ding{51}}%
\newcommand{\xmark}{\ding{55}}%

\newcommand{\firstvisit}[3]{\mathit{firstVisit}_{{#1}}(#2, #3)}
\newcommand{\lastvisit}[3]{\mathit{lastVisit}_{{#1}}(#2, #3)}

\newcommand{\proph}{\mathfrak{P}}
\newcommand{\pro}[2]{\xi_{{#1}, {#2}} }
\newcommand{\prov}[2]{p_{{#1}, {#2}} }
\newcommand{\proRel}[3]{\xi_{{#1}, {#2}, #3} }
\newcommand{\provRel}[3]{p_{{#1}, {#2}, #3} }

\newcommand{\lpromod}[3]{\mathfrak{P}^{\mathit{bin}}_{#2, #1, #3} }
\newcommand{\lpromin}[3]{\mathfrak{P}^{\mathit{min}}_{#2, #1, #3}}
\newcommand{\lpro}[2]{\mathfrak{P}_{{#1}, {#2}} }

\newcommand{\lproRel}[3]{\mathfrak{P}_{#1, #2, {#3}} }
\newcommand{\lproRell}[3]{\mathfrak{P}_{#1, #2, {#3}} }
\newcommand{\lproRab}[4]{\mathfrak{P}_{#1, #2, #3, {#4}} }

\newcommand{\ap}[0]{\mathit{AP}}

\newcommand{\vertex}{\mathbf{v}}
\newcommand{\sSeq}[1]{s_\vertex(#1)}
\newcommand{\ssSeq}[1]{s'_\vertex(#1)}

\newcommand{\bSeq}[1]{A_\vertex(#1)}

\newcommand{\sTr}[1]{t_\vertex(#1)}
\newcommand{\ssTr}[1]{t'_\vertex(#1)}
\newcommand{\trs}{t_\vertex}
\newcommand{\trss}{t'_\vertex}

\newcommand{\qSeqHat}[1]{\hat{q}_\vertex(#1)}

\newcommand{\veri}{\mathfrak{V}}
\newcommand{\refu}{\mathfrak{R}}

\newcommand{\aut}{\mathcal{A}}
\newcommand{\tss}{\mathcal{T}}

\newcommand{\ltlN}{\LTLnext}
\newcommand{\ltlG}{\LTLglobally}
\newcommand{\ltlF}{\LTLeventually}
\newcommand{\ltlU}{\LTLuntil}

\newcommand{\zip}[1]{\mathit{zip}(#1)}

\newcommand{\tracevars}{\mathcal{V}}
\newcommand{\traceSet}{\mathbb{T}}

\newcommand{\sucs}[1]{\mathit{Sucs}(#1)}
\newcommand{\tssto}{\xrightarrow{\tss}}

\newcommand{\traces}[1]{\mathit{Traces}(#1)}
\newcommand{\paths}[1]{\mathit{Paths}(#1)}
\newcommand{\game}[2]{\mathcal{G}_{#1, #2}}

\newcommand{\winss}[2]{#1 \vdash #2}

\newcommand{\calL}{\mathcal{L}}

\newcommand{\calO}{\mathcal{O}}

\newcommand{\gamenode}[1]{\langle #1 \rangle}
\newcommand{\reponse}[3]{\mathit{T}(#1, #2, #3)}
\newcommand{\optPath}[3]{\mathit{opt}(#1, #2, #3)}

\newcommand{\nat}{\mathbb{N}}

\newcommand{\refApp}[1]{Appendix \ref{#1}}

\algsetblockdefx[Loopi]{RepeatOne}{}
{1}{\algorithmicindent}
[0]{\textbf{repeat}}

\newcommand{\ldot}{\mathpunct{.}}

\newcommand{\forallu}{\accentset{\sim}{\forall}}
\newcommand{\existsu}{\accentset{\sim}{\exists}}

\newcommand{\decode}[1]{\mathit{dec}(#1)}
\newcommand{\encode}[1]{\mathit{enc}(#1)}

\newfloat{strat}{htbp}{lok}
\floatname{strat}{Strategy}

\crefname{strat}{strategy}{strategies}
\Crefname{strat}{Strategy}{Strategies}

\makeatletter
\newcommand\fs@plainruled{\def\@fs@cfont{\rmfamily}\let\@fs@capt\floatc@plain
	\def\@fs@pre{\hrule height1pt depth0pt \kern5pt}%
	\def\@fs@post{}%
	\def\@fs@mid{\kern2pt\hrule height1pt depth0pt\relax\kern\abovecaptionskip}%
	\let\@fs@iftopcapt\iffalse}
\makeatother

\floatstyle{plainruled}
\restylefloat{algorithm}
\restylefloat{strat}

\makeatletter
\newenvironment{alg}[1][htb]{%
	\renewcommand{\ALG@name}{Alg.}
	\begin{algorithm}[#1]%
	}{\end{algorithm}}
\makeatother

\newcommand\copyrighttext{%
	\footnotesize \textcopyright 2022 IEEE. Personal use of this material is permitted. Permission from IEEE must be obtained for all other uses, in any current or future media, including reprinting/republishing this material for advertising or promotional purposes, creating new collective works, for resale or redistribution to servers or lists, or reuse of any copyrighted component of this work in other works.}
\newcommand\copyrightnotice{%
	\begin{tikzpicture}[remember picture,overlay]
		\node[anchor=south,yshift=10pt] at (current page.south) {\fbox{\parbox{\dimexpr\textwidth-\fboxsep-\fboxrule\relax}{\copyrighttext}}};
	\end{tikzpicture}%
}

\newif\iffullversion
\fullversiontrue

\newcommand{\ifFull}[2]{\iffullversion#1\else#2\fi}

\title{Prophecy Variables for Hyperproperty Verification}

\author{\IEEEauthorblockN{Raven Beutner} \IEEEauthorblockA{CISPA Helmholtz Center for \\Information Security\\ Germany} \and \IEEEauthorblockN{Bernd Finkbeiner} \IEEEauthorblockA{CISPA Helmholtz Center for \\Information Security\\ Germany}}

\begin{document}

\maketitle

\begin{abstract}
Temporal logics for hyperproperties like HyperLTL use trace quantifiers to express properties that relate multiple system runs. 
In practice, the verification of such specifications is mostly limited to formulas without quantifier alternation, where verification can be reduced to checking a trace property over the self-composition of the system. 
Quantifier alternations like $\forall \pi. \exists \pi'. \phi$, can either be solved by complementation or with an interpretation as a two-person game between a $\forall$-player, who incrementally constructs the trace $\pi$, and an $\exists$-player, who constructs $\pi'$ in such a way that $\pi$ and $\pi'$ together satisfy $\phi$. 
The game-based approach is significantly cheaper but incomplete because the $\exists$-player does not know the future moves of the $\forall$-player. 
In this paper, we establish that the game-based approach can be made complete by adding ($\omega$-regular) temporal prophecies.
Our proof is constructive, yielding an effective algorithm for the generation of a complete set of prophecies.
\end{abstract}

\begin{IEEEkeywords}
Hyperproperties, HyperLTL, Hyperliveness, Verification, Prophecy Variables, Completeness 
\end{IEEEkeywords}

\iffullversion
\copyrightnotice
\fi

\section{Introduction}\label{sec:intro}

Hyperproperties \cite{ClarksonS08} are system properties that relate multiple execution traces in a system and commonly arise, e.g., in information-flow policies.
An increasingly popular logic for the specification of general hyperproperties is HyperLTL \cite{ClarksonFKMRS14}, which extends linear-time temporal logic (LTL) with explicit trace quantification.
In HyperLTL we can, for example, express a simple variant of non-interference (NI)~\cite{RoscoeWW96} as follows:
\begin{align*}
    \forall \pi. \forall \pi'\ldot \ltlG \Big( \bigwedge_{a \in L_\mathit{in}} a_{\pi} \leftrightarrow a_{\pi'} \Big) \to \ltlG \Big( \bigwedge_{a \in L_\mathit{out}} a_{\pi} \leftrightarrow a_{\pi'} \Big)
\end{align*}
Here $L_\mathit{in}$ and $L_\mathit{out}$ are sets of atomic propositions denoting low-security inputs and outputs. 
Sets $H_\mathit{in}$ and $H_\mathit{out}$ are the high-security counterparts.
The HyperLTL property states that any two traces with identical low-security inputs have identical low-security outputs, i.e., the system behaves deterministically for a low-security user.
A less strict notation of non-interference, in the literature often referred to as \emph{generalized non-interference} (GNI) \cite{McCullough88}, can be expressed as follows:
\begin{align*}
    \forall \pi. \forall \pi'. \exists \pi''\ldot \ltlG \Big( \!\! \bigwedge_{a \in L_\mathit{in} \cup L_\mathit{out}} \!\!\!\!\!a_{\pi} \leftrightarrow a_{\pi''} \Big) \land \ltlG \Big( \!\! \bigwedge_{a \in H_\mathit{in}} \!\!a_{\pi'} \leftrightarrow a_{\pi''} \Big)
\end{align*}
GNI states that for all traces $\pi$ and $\pi'$, there \emph{exists} a third trace $\pi''$ that agrees with the low-security inputs and outputs of $\pi$ but with the high-security inputs of $\pi'$.
Phrased differently, any input-output behavior observable by a low-security user is compatible with any sequence of high-security inputs.
GNI is of particular interest as it applies to non-deterministic systems where the simple variant of NI is violated when the nondeterminism influences the low-security output.

In this paper, we study the verification of HyperLTL, i.e., the question of whether a given system satisfies a given property.
For HyperLTL, the structure of the quantifier prefix has direct implications on the complexity of the verification problem.
For our example properties, the fundamental difference (w.r.t.~verification) between NI and GNI, is that NI uses only universal quantification over traces (we say NI is alternation-free) whereas GNI involves a quantifier alternation.
Verification of alternation-free properties is well understood and is reducible to the verification of a trace property on a suitable self-composition of the system \cite{BartheDR11,FinkbeinerRS15}. 
By contrast, verification of properties involving alternations is much more challenging.
In the complementation-based approach  \cite{FinkbeinerRS15} a quantifier alteration like  $\forall \pi. \exists \pi'. \phi$ is interpreted as $\forall \pi. \neg \forall \pi'. \neg \phi$ which can be checked by incrementally eliminating quantifiers with interposed system complementation.
This complementation is infeasible for larger systems.

\subsection{Strategy-based Verification}

A first scalable verification method for $\forall^*\exists^*$ HyperLTL properties (i.e., properties that involve an arbitrary number of universal quantifiers followed by an arbitrary number of existential quantifiers, such as GNI) has been proposed by Coenen et al. \cite{CoenenFST19}, which we call \emph{strategy-based verification}.
The key idea is to interpret a $\forall \pi. \exists \pi'. \phi$ formula as a game.
The $\forall$-player controls the universally quantified trace by moving through the system (thereby producing a trace $\pi$) while the $\exists$-player reacts with moves in a separate copy of the system (thereby producing a trace $\pi'$). 
The $\exists$-player wins if $\pi$ combined with $\pi'$ satisfies $\phi$. 
The resulting verification approach is sound (i.e., a winning strategy for the $\exists$-player implies that the property holds) and much cheaper than the complementation-based method (the game can be solved in polynomial time whereas the complementation incurs an exponential blow-up). 
The method is, however, incomplete. 
The $\exists$-player can, in step $i$, only react to the moves of the $\forall$-player up to step $i$ (i.e., only a finite prefix of the trace constructed by the $\forall$-player) and has no access to future behavior.
See \Cref{sec:overview} for examples.

\subsection{Prophecies to the Rescue}

A common proof technique to make information about future events accessible are prophecy variables \cite{AbadiL91}.
In the context of hyperproperty verification, a prophecy provides the $\exists$-player with information about the future behavior of the $\forall$-player. 
Appropriately chosen prophecies result in the existence of a winning strategy for the $\exists$-player (who, in each step, has access to the prophecies), even in cases where there is no winning strategy without the prophecies \cite{CoenenFST19}. 
However, in the context of hyperproperty verification, prophecies have, so far, been used as an \emph{ad hoc} method where prophecies are provided by the user on a case-by-case basis \cite{CoenenFST19}.
With this paper, we conduct a first formal study into the expressive power of prophecies.
In particular, we show that ($\omega$-regular) prophecies are complete, i.e., prophecies always suffice for successful verification.
Our main result informally reads as follows:

\begin{center}
    \parbox{0.9\columnwidth}{\textit{For any finite-state system $\tss$ and $\forall^*\exists^*$ HyperLTL property $\varphi$, there exist finitely many ($\omega$-regular) prophecies such that the $\exists$-player has a winning strategy (with access to the prophecies) if and only if $\tss$ satisfies $\varphi$.}}
\end{center}

\noindent
When given such a complete set of prophecies, verification of a hyperproperty reduces (in a sound-and-complete manner) to solving a finite-state two-player game.
Notably, our proof of the above result is constructive, i.e., we give an explicit (and effective) construction of a complete set of prophecies, represented as $\omega$-automata. 

\subsection{Prototype Implementation}

We have implemented our prophecy construction in a prototype model checker for $\forall^*\exists^*$ HyperLTL formulas, called \texttt{HyPro} (short for \textbf{Hy}perproperty Verification with \textbf{Pro}phecies).
If required, \texttt{HyPro} automatically constructs a complete set of prophecies and thus constitutes the first \emph{complete} verifier for $\forall^*\exists^*$ HyperLTL formulas with a safety matrix (see \Cref{sec:eval}). 
We emphasize that this paper's main contribution is a completeness proof for prophecies in hyperproperty verification.
While \texttt{HyPro} demonstrates that our explicit prophecy construction is applicable in practice, it is, currently, limited to small systems.

\subsection{Structure}

The remainder of this paper is structured as follows.
In \Cref{sec:overview} we demonstrate the need for prophecies on a small example and outline our automatic prophecy construction.
In \Cref{sec:relatedWork} we discuss related approaches, and in \Cref{sec:prelim} define preliminaries and introduce HyperLTL. 
We define strategy-based verification and prophecies in \Cref{sec:gameBased}, and discuss completeness in \Cref{sec:comp}.
Afterward, we first outline our prophecy construction for HyperLTL specifications where the matrix is a safety property (in \Cref{sec:compSafety}), and then extend it to full $\omega$-regularity in \Cref{sec:compGeneral}. 
In \Cref{sec:eval} we discuss prophecy-based verification and evaluate our prototype model checker \texttt{HyPro}.
Lastly, we outline further applications of (and future directions for) prophecy-based verification (in \Cref{sec:conclusionFuture}).

\section{Overview}\label{sec:overview}

In this section, we demonstrate the need for prophecies in hyperproperty verification on two small examples (in \Cref{sec:runningExample,sec:example}). Afterward, we sketch our automated prophecy construction (in \Cref{sec:overviewConstruction}).

\subsection{Strategy-based Verification and Prophecies}\label{sec:runningExample}

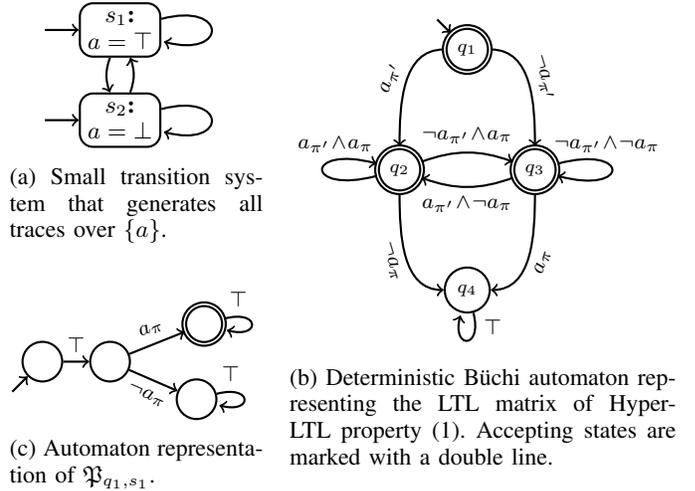
\begin{figure}
	\begin{minipage}{0.38\columnwidth}
		\begin{subfigure}{\linewidth}
			\centering
			\small
			\begin{tikzpicture}
				\node[draw, rectangle, thick,align=center, rounded corners=5pt] at (0,0) (n0) {$s_1$\textbf{:}\\$a = \top$};
				\node[draw, rectangle, thick,align=center, rounded corners=5pt] at (0,-1.2) (n1) {$s_2$\textbf{:}\\$a = \bot$};

				\draw[->, thick] (n0)+(-1,0) -- (n0);
				
				\draw[->, thick] (n1)+(-1,0) -- (n1);
				
				\path (n0) edge [thick, ->, bend right=20] (n1);
				\path (n1) edge [thick, ->, bend right=20] (n0);
				
				\path (n0) edge [loop right, thick, ->] (n0);
				
				\path (n1) edge [loop right, thick, ->] (n1);
			\end{tikzpicture}
			
			\subcaption{Small transition system that generates all traces over $\{a\}$.}\label{fig:ilEx1}
		\end{subfigure}
		
		\vspace{5mm}
		
		\begin{subfigure}{\linewidth}
			\centering
			\begin{tikzpicture}
				\node[circle, draw, thick,minimum size=15pt] at (-0.9,0) (ns) {};
				\node[circle, draw, thick,minimum size=15pt] at (0,0) (n0) {};
				\node[circle, draw, thick,minimum size=15pt] at (1.15,-0.5) (n1) {};
				\node[circle, draw,double, thick,minimum size=15pt] at (1.25,0.5) (n2) {};
				
				\draw[->,thick] (ns) to node[sloped,above] {\footnotesize$\top$} (n0);
				
				\draw[->,thick] (n0) to node[sloped,above] {\footnotesize$a_\pi$} (n2);
				\draw[->,thick] (n0) to node[sloped,below] {\footnotesize$\neg a_\pi$} (n1);
				
				\path (n1) edge [loop right,thick] node[above,yshift=3pt, xshift=-5pt] {\footnotesize$\top$} (n1);
				\path (n2) edge [loop right,thick] node[above,yshift=3pt, xshift=-5pt] {\footnotesize$\top$} (n1);
				
				\draw[->,thick] (ns)+(-0.4, -0.4) -- (ns);
			\end{tikzpicture}
			\vspace{-4mm}
			\addtocounter{subfigure}{1}
			\subcaption{Automaton representation of $\lpro{q_1}{s_1}$.}\label{fig:ilEx3}
		\end{subfigure}
	\end{minipage}\hfill
	\begin{minipage}{0.58\columnwidth}
		\begin{subfigure}{\linewidth}
			\centering
			\small
			\begin{tikzpicture}[scale=1.0]
				\node[draw, circle, thick,double,minimum size=15pt] at (0,-1) (n1) {\scriptsize$q_1$};
				
				\node[draw, circle, thick,double,minimum size=15pt] at (-0.9,-2.6) (n2) {\scriptsize$q_2$};
				
				\node[draw, circle, thick,double,minimum size=15pt] at (0.9,-2.6) (n3) {\scriptsize$q_3$};
				
				\node[draw, circle, thick,minimum size=15pt] at (0,-4.2) (n4) {\scriptsize$q_4$};

				\draw[->, thick] (n1)+(-0.4,0.4) -- (n1);
				
				\path (n1) edge [thick, ->,out=180, in =90] node[above, sloped] {\footnotesize$a_{\pi'}$} (n2);
				\path (n1) edge [thick, ->,out=0, in = 90] node[above, sloped] {\footnotesize$\neg a_{\pi'}$} (n3);

				\path (n2) edge [thick, ->, loop left, min distance=10mm,looseness=10] node[above,yshift=3pt, xshift=5pt] {\footnotesize$ a_{\pi'} \! \land  \! a_\pi$} (n2);
				
				\path (n3) edge [thick, ->, loop right,min distance=10mm,looseness=10] node[above,yshift=3pt, xshift=-2pt] {\footnotesize$\neg a_{\pi'} \! \land \! \neg a_\pi$} (n3);

				\path (n3) edge [thick, ->, bend left = 20] node[sloped, below] {\footnotesize$a_{\pi'} \!\land \!\neg a_\pi$} (n2);
				
				\path (n2) edge [thick, ->, bend left = 20] node[sloped, above] {\footnotesize$\neg a_{\pi'}\! \land \!a_\pi$} (n3);

				\path (n2) edge [thick, ->,out=-90, in=180] node[sloped, below] {\footnotesize$\neg  a_\pi$} (n4);
				\path (n3) edge [thick, ->,out=-90, in=0] node[sloped, below] {\footnotesize$a_\pi$} (n4);

				\path (n4) edge [thick, ->,loop below] node[right, xshift=3pt,yshift=5pt] {\footnotesize$\top$} (n4);
			\end{tikzpicture}
			\addtocounter{subfigure}{-2}
			\subcaption{Deterministic Büchi automaton representing the LTL matrix of HyperLTL property (\ref{eq:exampleProp}). Accepting states are marked with a double line. }\label{fig:ilEx2}
		\end{subfigure}
	\end{minipage}

	\caption{A simple example that demonstrates why prophecies are needed for successful strategy-based verification. \Cref{fig:ilEx3} depicts the (minimized) automaton resulting from our prophecy construction. }
\end{figure}

As a (very) small example, consider the transition system $\tss$ in \Cref{fig:ilEx1}, which generates all traces over atomic propositions $\ap = \{a\}$, and the HyperLTL specification  
\begin{align}\label{eq:exampleProp}
	\varphi \coloneqq \forall \pi. \exists \pi'. \ltlG (a_{\pi'} \leftrightarrow \ltlN a_\pi).
\end{align}
The property states that for every trace $\pi$ there should be a trace $\pi'$ that mimics $\pi$ \emph{one step into the future}.
Clearly $\tss \models \varphi$, i.e., the system satisfies the property. 

To automatically check this using strategy-based verification \cite{CoenenFST19}, we construct a game where, in each step, the $\forall$-player chooses a successor state for trace $\pi$ (in the first step the $\forall$-player chooses any initial state), and the $\exists$-player reacts by choosing a successor state for trace $\pi'$ (in a separate copy of the system). 
The $\exists$-player tries to construct trace $\pi'$ such that $\pi'$ combined with $\pi$ satisfies the LTL matrix of (\ref{eq:exampleProp}).
However, even though $\tss \models \varphi$, the $\exists$-player loses this game.
In every step of the game, the $\exists$-player needs to move to either $s_1$ or $s_2$.
With either choice, the $\forall$-player can (in the next step of the game) move its copy to the opposite state (i.e., move to $s_2$ if the $\exists$-player moved to $s_1$ and vice versa) and thereby ensure that $a_{\pi'} \not\leftrightarrow \ltlN a_\pi$ holds;
strategy-based verification fails.

To win the game, the $\exists$-player would need to base its decision on the \emph{next} move of the $\forall$-player.
Prophecies can provide this necessary information about the future behavior of the $\forall$-player.
Consider the LTL-definable prophecy $\xi \coloneqq \ltlN a_\pi$.\footnote{In our setting, a prophecy is a $\omega$-regular set of behaviors of the universally quantified traces. 
	If possible, we can represent this set as an LTL formula. }
If the $\exists$-player has access to this prophecy (i.e., has access to an oracle that tells him, in each step of the game, if $\xi$ currently holds or not), a winning strategy exists. 
For example, if $\xi$ holds (so $a$ holds in the next step on $\pi$), the $\exists$-player moves to $s_1$ as this ensures $a_{\pi'} \leftrightarrow \ltlN a_\pi$.

\begin{figure}[t!]
	\begin{subfigure}{0.48\columnwidth}
		\centering
		\begin{algorithmic}[1]
			\RepeatOne
			\If{$\star$}\label{line:choice}
			\State $h \coloneqq \star$
			\State $o \coloneqq h$
			\Else
			\State $h \coloneqq \star$
			\State $o \coloneqq \neg h$
			\EndIf
		\end{algorithmic}
		\subcaption{Simple example program. Here, $\star$ denotes a non-deterministic choice of a (boolean) value. }\label{fig:introExample2}
	\end{subfigure}
	\hfil
	\begin{subfigure}{0.48\columnwidth}
		\centering
		\small
		\begin{tikzpicture}
			\node[draw, rectangle, thick,align=center, rounded corners=5pt] at (0,0) (n0) {$s_1$\textbf{:}\\$o = \bot$\\$h = \bot$};
			
			\node[draw, rectangle, thick,align=center, rounded corners=5pt] at (1.5,0) (n1) {$s_2$\textbf{:}\\$o = \bot$\\$h = \bot$};
			
			\node[draw, rectangle, thick,align=center, rounded corners=5pt] at (1.5,1.5) (n2) {$s_3$\textbf{:}\\$o = \bot$\\$h = \top$};
			
			\node[draw, rectangle, thick,align=center, rounded corners=5pt] at (1.5, -1.5) (n3) {$s_4$\textbf{:}\\$o = \top$\\$h = \bot$};
			
			\node[draw, rectangle, thick,align=center, rounded corners=5pt] at (-1.5,0) (n4) {$s_5$\textbf{:}\\$o = \bot$\\$h = \bot$};
			
			\node[draw, rectangle, thick,align=center, rounded corners=5pt] at (-1.5,1.5) (n5) {$s_6$\textbf{:}\\$o = \bot$\\$h = \bot$};
			
			\node[draw, rectangle, thick,align=center, rounded corners=5pt] at (-1.5, -1.5) (n6) {$s_7$\textbf{:}\\$o = \top$\\$h = \top$};
			
			\draw[->, thick] (n0)+(0,1) -- (n0);
			
			\draw[->, thick] (n0) -- (n1);
			\draw[->, thick] (n1) -- (n2);
			\draw[->, thick] (n1) -- (n3);
			\draw[->, thick] (n2) -- (n0);
			\draw[->, thick] (n3) -- (n0);
			
			\draw[->, thick] (n0) -- (n4);
			\draw[->, thick] (n4) -- (n5);
			\draw[->, thick] (n4) -- (n6);
			\draw[->, thick] (n5) -- (n0);
			\draw[->, thick] (n6) -- (n0);
		\end{tikzpicture}
		
		\subcaption{A simplified transition system obtained from the program in  \Cref{fig:introExample2}.}\label{fig:introExample3}
	\end{subfigure}

	\caption{Simple example program that requires prophecies to successfully verify GNI using strategy-based verification.  }
\end{figure}
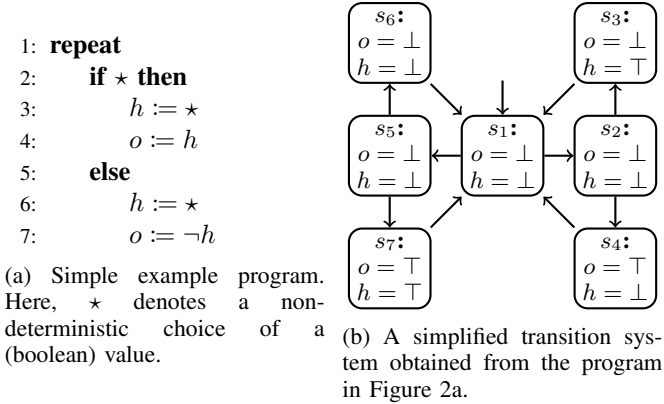

\subsection{Prophecies and GNI}\label{sec:example}

Prophecies are also needed when applying strategy-based verification to more realistic systems and properties.
As a second example, consider the program in \Cref{fig:introExample2} where $h$ is a high-security input and $o$ a low-security output, and the GNI property from \Cref{sec:intro}.
\Cref{fig:introExample3} depicts a simplified version of the program as a transition system.
In state $s_1$ the system can non-deterministically transition into $s_2$ or $s_5$. 
From $s_2$ the values of $h$ and $o$ disagree (as in the second branch of the conditional in \Cref{fig:introExample2}) whereas from $s_5$ the values agree (as in the first branch).
It is easy to see that this program (and/or transition system) satisfies GNI, but strategy-based verification fails.
In order to resolve the non-deterministic choice in line \ref*{line:choice} of the program (or in state $s_1$ of the transition system), the $\exists$-player needs to know the \emph{next} input and output on traces $\pi$ and $\pi'$.
Prophecies can provide the needed information about the future behavior on $\pi, \pi'$.
This information is, for example, made available via the LTL-definable prophecies $\xi_1 \coloneqq \ltlN o_{\pi}$ and $\xi_2 \coloneqq  \ltlN h_{\pi'}$.
With access to these prophecies, a winning strategy for the $\exists$-player exists. 
For example, if $\xi_1$ holds (so the next value of $o$ on $\pi$ is $\top$) and $\xi_2$ does not hold (so the next value of $h$ on $\pi'$ is $\bot$), the $\exists$-player moves to $s_2$ as this supports a later transition to the state $s_4$ (where $o = \top, h = \bot$ as required).
Note that, different from the example in \Cref{sec:runningExample}, this dependency on the next state is \emph{not} explicit in the property (as GNI does not involve any $\ltlN$s).

\subsection{Automated Prophecy Construction}\label{sec:overviewConstruction}

We now sketch how to \emph{automatically} construct a complete set of prophecies, i.e., a set of prophecies that ensures that the $\exists$-player can win the game (provided the property holds).

As a concrete example, we use the system and property from \Cref{sec:runningExample}.
For this example, the LTL matrix of (\ref{eq:exampleProp}) is a safety property, which simplifies the prophecy construction significantly. 
Conceptually, the idea is to design prophecies that directly identify those states that the $\exists$-player can move to without losing the game.
We observe that the $\exists$-player can (safely) move to state $s$ (for $s \in \{s_1, s_2\}$) iff the trace $\pi$ constructed by the $\forall$-player is such that there exists \emph{some} trace $\pi'$ starting in $s$ that serves as a witness for $\pi$.
To formalize this, \Cref{fig:ilEx2} depicts a deterministic Büchi automaton $\aut$ that tracks the matrix of (\ref{eq:exampleProp}).
Given an automaton state $q$ (for $q\in \{q_1,q_2, q_3, q_4\}$) and a system state $s$, we summarize all traces constructed by the $\forall$-player on which a witness trace starting in state $s$ exists (where the automaton begins tracking in state $q$).
Formally we define
\begin{align*}
	\lpro{q}{s}\coloneqq \{ t \in (2^{\{a\}})^\omega \mid \exists t' \in \mathit{Traces}(\tss_{s})\ldot  t \otimes t' \in \calL(\aut_q)  \}
\end{align*}
where $\mathit{Traces}(\tss_s)$ are all traces starting in state $s$, $\calL(\aut_q)$ are all traces accepted by $\aut$ starting in state $q$, and $t \otimes t'$ is the pointwise product of $t$ and $t'$. 
See \Cref{sec:compSafety} for a formal treatment.

The resulting prophecies determine which move is safe for the $\exists$-player:
If during the game the current state of the automaton tracking the matrix of (\ref{eq:exampleProp}) is $q$, and prophecy $\lpro{q}{s}$ holds (i.e., the trace constructed by the $\forall$-player is contained in this set), the $\exists$-player can safely pick $s$ as its successor.

As an example, we consider the set $\lpro{q_1}{s_1}$.
By taking the product of $\tss$ and $\aut$, we obtain an automaton representation of $\lpro{q_1}{s_1}$, which, after minimization, results in the Büchi automaton depicted in \Cref{fig:ilEx3}.
Coincidentally, this automaton directly corresponds to the LTL prophecy $\xi \coloneqq \ltlN a_\pi$ identified in \Cref{sec:runningExample}.
As we already argued in \Cref{sec:runningExample}, the single prophecy $\lpro{q_1}{s_1}$ thus provides sufficient information for the $\exists$-player.\footnote{In general, our completeness result (for cases where the matrix of the HyperLTL formula is a safety property) states that the set $\{\lpro{q}{s}\}_{q \in Q, s \in S}$ where $Q$ is the set of automaton states and $S$ the set of system states \emph{always} provides sufficient information for the $\exists$-player to win (provided the property holds). For properties where the matrix does not denote a safety property, a more involved construction is necessary (see \Cref{sec:compGeneral}). }

\section{Related Work}\label{sec:relatedWork}

\subsubsection{Hyperproperty Verification}
Recently, the automated verification of hyperproperties expressed in general logics has received significant attention.
Verification of alternation-free formulas (and, in particular, $k$-safety) is reducible to the verification of a trace property on the self-composition of the system \cite{FinkbeinerRS15,BartheDR11}.
In contrast, few attempts at the automatic verification of properties involving a quantifier alteration have been made.
This is in stark contrast to the fact that many relevant properties (especially in non-deterministic systems) require alternation.
Examples include information-flow policies like GNI, refinement properties, fairness, and robust cleanness.
Barthe et al.~\cite{BartheCK13} describe an asymmetric product of the system such that only a subset of the behavior of the second system is preserved, thereby allowing the verification of $\forall\exists$ properties. 
It is challenging to construct an asymmetric product and verify its correctness (i.e., show that the product preserves all behavior of the first, universally quantified, system). 
Unno et al.~\cite{UnnoTK21} describe a constraint-based approach to verify functional (opposed to temporal) $\forall\exists$ properties.
In their framework, both the existentially quantified traces and the scheduling of the system are encoded in an extension of constraint Horn clauses. 
Lamport and Schneider \cite{LamportS21} outline a deductive approach to verify hyperproperties by reducing the verification to TLA.
This is possible as existential trace quantification can be internalized into the TLA specification.
Hsu et al.~\cite{HsuSB21} present a bounded model checking algorithm for hyperproperties.
As usual for bounded approaches, a property can only be refuted if there exists a finite set of finite paths refuting it; bounded model checking for hyperproperties is incomplete. 
A first practical (albeit incomplete) algorithm for the verification of temporal properties involving quantifier alternation (expressed in HyperLTL) was proposed by Coenen et al.~\cite{CoenenFST19} in the form of strategy-based verification, which forms the basic setting of this work. 
Strategy-based verification is also applicable to infinite-state systems \cite{BeutnerF22CAV}.

\subsubsection{Prophecy Variables}

Abadi and Lamport have introduced the concept of prophecies as a proof technique in the context of refinement mappings between state machines, and have shown completeness in this setting \cite{AbadiL91}.
Coenen et al.~\cite{CoenenFST19} use prophecies to strengthen the $\exists$-player in strategy-based verification.
It is important to note that the use of temporal prophecies advocated in \cite{CoenenFST19} (and studied in this paper) differs from the setting of Abadi and Lamport \cite{AbadiL91} in several key regards.  
In \cite{AbadiL91}, a prophecy variable changes the system by adding a variable that records the future behavior of the system as a sequence of states.\footnote{In particular, the completeness proof in \cite{AbadiL91} is purely semantic. The history and prophecy variables describe the past and future behavior of the system, which, in the worst case, turns a finite-state system into an infinite-state one.}
We take a different point of view: In our setting, we do not manipulate the system but define a prophecy as a $\omega$-regular set of behavior (expressed in temporal logic). 
The $\exists$-player is only provided with a single bit of information that indicates if the future behavior of $\forall$-player lies within the prophecy or not.

While Coenen et al.~\cite{CoenenFST19} already discuss prophecies, they consider them as an \emph{ad hoc} feature where the user must provide prophecies on a case-by-case basis.
We study prophecies in the same setting (albeit our prophecies are $\omega$-regular and not necessarily LTL-definable as in \cite{CoenenFST19}) but conduct a systematic analysis of the expressiveness of strategy-based verification when enriched with prophecies.
In particular, we establish that prophecies \emph{always} suffice to verify a property and give an explicit (and fully automatic) algorithm for the construction of a complete set of prophecies. 
Compared to the purely semantic construction of Abadi and Lamport \cite{AbadiL91}, we work in the fixed framework of $\omega$-regularity and represent prophecies as $\omega$-automata.

Prophecies as a proof technique have found application in various settings. 
They have been used for the verification of branching-time properties \cite{CookKP15}, the construction of simulations between automata \cite{LynchV95}, to strengthen proofs in program logics \cite{JungLPRTDJ20,ZhangFFSL12,Vafeiadis08}, and to construct liveness-to-safety transformations \cite{PadonHMPSS21}.
Cook and Koskinen \cite{CookK11} introduce prophecies in the form of decision predicates to verify LTL properties using CTL solvers on infinite-state systems.
A decision predicate can be seen as a limited form of (non-boolean) temporal prophecy that predicts the \emph{number} of occurrences of a particular event in the future. 
Closely related to our setting is the work by Unno et al.~\cite{UnnoTK21}.
They show that for the verification of functional $\forall\exists$ properties, it is sufficient to have a prophecy variable that simply predicts the final state of the universally quantified execution. 
In our temporal setting, the prophecy construction is necessarily more complex as it needs to provide information about the temporal behavior of the universally quantified execution, and the information communicated per prophecy is restricted to a single bit.

\section{Preliminaries}\label{sec:prelim}

We fix a set of atomic propositions $\ap$ and define $\Sigma \coloneqq 2^\ap$.
A trace is an element $t \in \Sigma^\omega$. 
We write $t(i)$ to denote the $i$th element (starting with $0$) and $t[i,\infty]$ for the infinite suffix starting at position $i$.
For traces $t_1, \ldots, t_n \in \Sigma^\omega$ we define $\zip{t_1, \ldots, t_n} \in (\Sigma^n)^\omega$ as the pointwise product of the traces, i.e.,  $\zip{t_1, \ldots, t_n}(i) \coloneqq (t_1(i), \ldots, t_n(i))$. 
We occasionally write $t_1 \otimes t_2$ instead of $\zip{t_1, t_2}$.

\subsubsection{Transition Systems} 

A transition system is a tuple $\tss= (S, S_0, \varrho, L)$ where $S$ is a finite set of states, $S_0 \subseteq S$ a set of initial states, $\varrho \subseteq S \times S$ a transition relation, and $L : S \to \Sigma$ a state labelling. 
We write $s \xrightarrow{\tss} s'$ whenever $(s, s') \in \varrho$ and define $\sucs{s} \coloneqq \{s' \mid  (s, s') \in \varrho \}$.
We assume $\sucs{s} \neq \emptyset$ for every $s \in S$.
A path in $\tss$ is an infinite sequence $p \in S^\omega$ such that $p(0) \in S_0$ and for every $i \in \mathbb{N}$, we have $p(i+1) \in \sucs{p(i)}$.
Each path $p$ denotes a trace $L(p) \in \Sigma^\omega$ by applying the labelling pointwise, i.e., $L(p)(i) \coloneqq L(p(i))$.
We write $\paths{\tss}$ for the set of all paths and $\traces{\tss}$ for the set of all traces. 
For $s \in S$ we define $\tss_s$ as the transition systems obtained by changing the initial states to $\{s\}$. 

\subsubsection{$\omega$-Automata}

A deterministic $\omega$-automaton over alphabet $\Sigma$ is a tuple $\mathcal{A} = (Q, q_0, \delta, \mathit{Acc})$ where $Q$ is a finite set of states, $q_0 \in Q$ an initial state, $\delta : Q \times \Sigma \to Q$ a transition function, and $\mathit{Acc} \subseteq Q^\omega$ the acceptance condition. 
For every finite word $u \in \Sigma^*$, we define $\delta^*(u) \in Q$ as the unique state reached when reading $u$ (starting in $q_0$).
For a trace $t \in \Sigma^\omega$, the unique run $r_t \in Q^\omega$ is given by $r_t(i) \coloneqq \delta^*(t[0, i-1])$ where $t[0, i-1]$ is the prefix of length $i$.
We write $\calL(\mathcal{A})$ for the language of the automaton, which consists of all traces $t$ whose unique run $r_t$ satisfies $r_t \in \mathit{Acc}$.
In a Büchi automaton, the acceptance is given by a set $F \subseteq Q$ of accepting states, and a run is accepting if it visits states in $F$ infinity many times. 
In a parity automaton, the acceptance is given by a coloring $c : Q \to \mathbb{N}$, and a run is accepting if the minimal color occurring infinitely often (as given by $c$) is even. 
In a safety automaton, the acceptance is given by a set $B \subseteq Q$ of bad states, and a run is accepting if it never visits a state in $B$.
A language $\calL \subseteq \Sigma^\omega$ is $\omega$-regular if there exists a deterministic parity automaton (DPA) that recognizes it.\footnote{Throughout this paper, we work with deterministic $\omega$-automata. Any non-deterministic Büchi automaton (NBA) (see, e.g., \cite{0020348} for a formal definition) can be effectively translated into a DPA \cite{Safra88,Piterman07}. 
	On the other hand, deterministic Büchi and deterministic safety automata are strictly less expressive and do not capture full $\omega$-regularity.
}
A language $\calL \subseteq \Sigma^\omega$ is \emph{safety} \cite{AlpernS85,KupfermanV99}, if it can be recognized by a deterministic safety automata.
Given $q \in Q$ we define $\aut_q$ as the automaton obtained by replacing the initial state with $q$.
For a set $X \subseteq Q$ and a trace $t \in \Sigma^\omega$, we define $\mathit{firstVisit}_X(\aut, t) \in \mathbb{N} \cup \{\infty\}$ as the first time step where the unique run of $\aut$ on $t$ visits a state in $X$ (if it exists and $\infty$ otherwise).

\subsubsection{Parity Games}

A parity game is a tuple $\mathcal{G} = (V_\veri, V_\refu,\allowbreak T, c)$ where $V \coloneqq V_\veri \mathbin{\mathaccent\cdot\cup} V_\refu$ is the finite set of states. 
The states in $V_\veri$ are controlled by the verifier $\veri$ and those in $V_\refu$ are controlled by the refuter $\refu$. $T \subseteq V \times V$ is the transition relation (we assume that for each $v$ there is at least one $v'$ with $(v, v') \in T$), and $c : V \to \mathbb{N}$ the coloring of each node.
 A strategy $\sigma$ for player $p \in \{\veri, \refu\}$ is a function $\sigma : V^* \times V_p \to V$  such that for every $\vertex \in V^*, v \in V_p$, $(v, \sigma(\vertex, v)) \in T$.
A play in $\mathcal{G}$ is an infinite sequence $r \in V^\omega$ such that for every $i$, $(r(i), r(i+1)) \in T$.
The play $r$ is compatible with strategy $\sigma$ for player $p$ if for every $i$ where $r(i) \in V_p$ we have that $r(i+1) = \sigma(r(0)\cdots r(i-1), r(i))$.
A play $r$ is won by player $\veri$ if the minimal color occurring infinitely often in $r$ (according to $c$) is even. Otherwise, it is won by $\refu$. 
We say that player $p$ wins node $v$ if there exists a strategy $\sigma$ for $p$ such that every play that starts in $v$ and is compatible with $\sigma$ is won by $p$. 
As parity games are positionally determined \cite{martin1975borel}, every node is either won by $\veri$ or by $\refu$.

\subsubsection{HyperLTL}
As the basic specification language for hyperproperties we use HyperLTL \cite{ClarksonFKMRS14}, which extends linear-time temporal logic (LTL) with explicit trace quantification. 
We assume a fixed set of trace variables $\tracevars$.
Formulas in HyperLTL are generated by the following grammar.
\begin{align*}
    \varphi &\coloneqq \exists \pi. \varphi \mid \forall \pi. \varphi \mid \phi \\
    \phi &\coloneqq a_\pi \mid \neg \phi \mid \phi_1 \land \phi_2 \mid \ltlN \phi \mid \phi_1 \ltlU \phi_2
\end{align*}%
where $\pi \in \tracevars$ and $a \in \ap$.
We use the derived boolean connectives $\lor, \to, \leftrightarrow$, boolean constants $\top, \bot$, and temporal operators eventually ($\ltlF \phi \coloneqq \top \ltlU \phi$) and globally ($\ltlG \phi \coloneqq \neg \ltlF \neg \phi$).
We consider only closed formulas, i.e., formulas where for each atom $a_\pi$ the trace variable $\pi$ is bound by some trace quantifier.
The semantics of HyperLTL is given with respect to a set of traces $\traceSet \subseteq \Sigma^\omega$ and a trace assignment $\Pi$, which is a partial mapping $\Pi : \tracevars \rightharpoonup \Sigma^\omega$.
For $\pi \in \tracevars$ and trace $t$, we write $\Pi[\pi \mapsto t]$ for the trace assignment obtained by updating the value of $\pi$ to $t$.
\begin{align*}
    \Pi, i &\models  a_\pi &\text{ iff } \quad  &a \in \Pi(\pi)(i)\\
     \Pi, i &\models  \neg \phi &\text{ iff }\quad & \Pi, i \not\models  \phi \\
    \Pi, i &\models  \phi_1 \land \phi_2 &\text{ iff }\quad  &\Pi, i \models \phi_1 \text{ and }  \Pi, i \models  \phi_2\\
    \Pi, i&\models  \LTLnext  \phi &\text{ iff }\quad & \Pi, i + 1 \models \phi \\
    \Pi, i&\models  \phi_1 \ltlU \phi_2 &\text{ iff } \quad& \exists j \geq i\ldot \Pi, j\models  \phi_2 \text{ and } \\
    &&&\quad\forall i \leq k < j\ldot  \Pi, k \models  \phi_1\\
    \Pi &\models_\traceSet  \phi &\text{ iff }\quad &\Pi, 0 \models \phi\\
    \Pi &\models_\traceSet \exists \pi. \varphi &\text{ iff }\quad &\exists t \in \traceSet\ldot \Pi[\pi \mapsto t] \models_{\mathbb{T}}  \varphi\\
    \Pi &\models_\traceSet  \forall \pi. \varphi &\text{ iff }\quad &\forall t \in \traceSet\ldot \Pi[\pi \mapsto t] \models_{\mathbb{T}}  \varphi
\end{align*}%
We say a transition system $\tss$ satisfies $\varphi$, written $\tss \models \varphi$, if $\emptyset \models_{\mathit{Traces}(\tss)} \varphi$ where $\emptyset$ denotes the empty trace assignment.

\begin{figure*}[!t]
	\begin{prooftree}
		\AxiomC{$s_1 \tssto s_1'$}
		\AxiomC{$\cdots$}
		\AxiomC{$s_k \tssto s_k'$}
		\AxiomC{$q' = \delta^\phi\Big(q, \big(L(s_1), \cdots, L(s_{k+l})\big) \Big)$}
		\RightLabel{\footnotesize($\forall$)}
		\QuaternaryInfC{$\gamenode{(s_1, \ldots, s_k, s_{k+1}, \ldots, s_{k+l}), q, \forall} \to  \gamenode{(s'_1, \ldots, s'_k, s_{k+1}, \ldots, s_{k+l}), q', \exists}$}
	\end{prooftree}
	
	\begin{prooftree}
		\AxiomC{$s_{k+1} \tssto s_{k+1}'$}
		\AxiomC{$\cdots$}
		\AxiomC{$s_{k+l} \tssto s_{k+l}'$}
		\RightLabel{\footnotesize($\exists$)}
		\TrinaryInfC{$\gamenode{(s_1, \ldots, s_k, s_{k+1}, \ldots, s_{k+l}), q, \exists} \to  \gamenode{(s_1, \ldots, s_k, s'_{k+1}, \ldots, s'_{k+l}), q, \forall}$}
	\end{prooftree}
	
	\begin{prooftree}
		\AxiomC{$s_{k+1} \in S_0$}
		\AxiomC{$\cdots$}
		\AxiomC{$s_{k+l} \in S_0$}
		\RightLabel{\footnotesize(init)}
		\TrinaryInfC{$(s_1, \ldots, s_k) \to \gamenode{(s_1, \ldots, s_k, s_{k+1}, \ldots, s_{k+l}), q^\phi_0, \forall}$}
	\end{prooftree}
	
	\caption{Transition rules for the parity-game-based synthesis of winning strategies for the $\exists$-player. }\label{fig:transitions}
\end{figure*}

\subsubsection{Quantified Propositional Temporal Logic (QPTL)}

The prophecies we study in this paper are $\omega$-regular sets. 
LTL is limited to non-counting properties and can consequently not express arbitrary $\omega$-regular properties \cite{DiekertG08}. 
To nevertheless support prophecies on a syntactic level (where we represent prophecies as formulas instead of $\omega$-automata), we use Quantified Propositional Temporal Logic (QPTL) \cite{sistla1983theoretical}.
We assume a fresh set of propositional variables $\mathit{PV}$.
We define QPTL formulas by the following grammar.
\begin{align*}
    \phi \coloneqq  a_\pi \mid \neg \phi \mid \phi_1 \land \phi_2 \mid \ltlN \phi \mid \phi_1 \ltlU \phi_2 \mid \existsu q. \phi \mid q
\end{align*}
where $\pi \in \tracevars$, $a \in \ap$ and $q \in \mathit{PV}$. 
QPTL allows the quantification of a proposition variable $q$ using $\existsu q. \phi$ and to refer to the truth value of each propositional variable. 
We abbreviate $\forallu q. \phi \coloneqq \neg \existsu q. \neg \phi$. 
Note that we write $\existsu$ and $\forallu$ for propositional quantification to visually distinguish them from the trace quantifiers in HyperLTL. 
The semantics of QPTL is defined similarly to before with an additional mapping $\Delta : \mathit{PV} \rightharpoonup \mathbb{B}^\omega$ that handles propositional quantification (where $\mathbb{B} = \{\top, \bot\}$).
\begin{align*}
    \Pi, \Delta, i &\models  a_\pi &\text{ iff } \quad &a \in \Pi(\pi)(i)\\
    \Pi,\Delta, i &\models  \neg \phi &\text{ iff } \quad& \Pi,\Delta, i \not\models  \phi \\
    \Pi,\Delta, i &\models  \phi_1 \land \phi_2 &\text{ iff } \quad &\Pi,\Delta, i \models \phi_1 \text{ and }  \Pi,\Delta, i \models  \phi_2\\
    \Pi,\Delta, i&\models  \LTLnext  \phi &\text{ iff } \quad& \Pi,\Delta, i + 1 \models \phi \\
    \Pi,\Delta, i&\models  \phi_1 \ltlU \phi_2 &\text{ iff } \quad& \exists j \geq i\ldot \Pi,\Delta, j\models  \phi_2 \text{ and } \\
    & & &\quad\forall i \leq k < j\ldot  \Pi,\Delta, k \models  \phi_1\\
    \Pi,\Delta, i &\models  \existsu q. \phi &\text{ iff } \quad&\exists \tau \in \mathbb{B}^\omega\ldot\Pi, \Delta[q \mapsto \tau], i \models  \varphi\\
     \Pi, \Delta, i &\models  q &\text{ iff } \quad &\Delta(q)(i) = \top
\end{align*}%
The main advantage of QPTL (over LTL) stems from the following result:
\begin{theorem}[\cite{sistla1983theoretical}]\label{theo:QPTL}
    A language $\calL$ is $\omega$-regular if and only if it is definable in QPTL. 
\end{theorem}

\begin{example}\label{ex:qptlProperty}
	Take the property ``$a$ holds on trace $\pi$ in at least one even position''.
	While not expressible in LTL \cite{DiekertG08}, we can express it in QPTL as 
	\begin{align*}
		\existsu q\ldot q \land \ltlG (q \leftrightarrow \ltlN \neg q) \land \ltlF (a_\pi \land q).
	\end{align*}
\end{example}

In the remainder of this paper, we assume no particular familiarity with QPTL and only use it when absolutely necessary.
We resort to QPTL as a tool to express $\omega$-regular properties as formulas which allows us to treat prophecies at a syntactic level. 
Our prophecy construction itself is language-theoretic.

\section{Strategy-based Verification}\label{sec:gameBased}

The problem we are tackling in this paper is the following:
Given a transition system $\tss$ and a $\forall^*\exists^*$-HyperLTL property $\varphi$, check if $\tss \models \varphi$. 
A first practical verification approach was proposed by Coenen et al.~\cite{CoenenFST19}, which we refer to as \emph{strategy-based verification}.
The idea is to instantiate existential quantification with a strategy that incrementally constructs a trace by reacting to the moves of the $\forall$-player.
Coenen et al.~formalize the strategy as a finite state transducer that determines the next move of all existentially quantified copies. 
The automated synthesis of a strategy is then expressed as a SMT constraints.
We phrase the problem as a parity game which serves as an easier formal foundation to discuss our completeness results.

\subsection{Strategy-based Verification as a Parity Game}
The idea is that the parity game mimic the iterative trace construction of both players.
Assume we are given a system $\tss = (S, S_{0}, \varrho, L)$ and a HyperLTL formula 
\begin{align*}
    \varphi = \forall \pi_1 \ldots \pi_k. \exists \pi_{k+1} \ldots \pi_{k+l}\ldot \phi.
\end{align*}
We define a parity game $\game{\tss}{\varphi}$ as follows.
Let $\aut^\phi = (Q^\phi, q^\phi_0, \delta^\phi, c^\phi)$ be a deterministic parity automaton (DPA) over $\Sigma^{k+l}$ for $\phi$ that accepts exactly the zippings of traces that satisfy the formula, i.e.,  $[\pi_1 \mapsto t_1, \ldots, \pi_{k+l} \mapsto t_{k+l}] \models \phi$ if and only if $\zip{t_1, \ldots, t_{k+l}} \in \calL(\aut^\phi)$.
The construction of this automaton can be performed via a standard LTL to DPA translation (see, e.g., \cite{VardiW94,FinkbeinerRS15}).

The game $\game{\tss}{\varphi}$ comprises two node kinds:
Nodes are either of the form $(s_1, \ldots, s_k)$ where $s_i \in S$ for all $1 \leq i \leq k$ to encode the initial states of the universally quantified copies.
Or they are of the form $\gamenode{(s_1, \ldots, s_{k+l}), q, \flat}$ where $s_i \in S$ for all $1 \leq i \leq k+l$, $q \in Q^\phi$ and $\flat \in \{\forall, \exists\}$. 
Here $(s_1, \ldots, s_{k+l})$ gives the current state of all copies of $\tss$, $q$ is the current state of the DPA tracking $\phi$, and $\flat$ defines whether the universal ($\flat = \forall$) or existential ($\flat = \exists$) copies move next. 
Nodes of the form $\gamenode{(s_1, \ldots, s_{k+l}), q, \forall}$ are controlled by the refuter (who takes the role of the $\forall$-player), and nodes of the form $\gamenode{(s_1, \ldots, s_{k+l}), q, \exists}$ and $(s_1, \ldots, s_k)$ are controlled by the verifier (who takes the role of the $\exists$-player).
The transitions of the game are given in \Cref{fig:transitions}.
The ($\forall$) and ($\exists$)-transition rules are the game's main rules.
In the ($\forall$)-rule all universally quantified copies are updated by moving to successor states within $\tss$.
Simultaneously, we update the automaton state of $\aut^\phi$. 
Similarly, in the ($\exists$)-rule, the existentially quantified copies are updated. 
The (init)-rule is used at the beginning where the universal copies have already chosen a state and the existential copies can select any initial state for themself.
Lastly, the coloring of the nodes is obtained by assigning each node of the form $\gamenode{(s_1, \ldots, s_{k+l}), q, \flat}$ the color given by $c^\phi(q)$. 
The color of nodes of the form $(s_1, \ldots, s_k)$ is irrelevant as they are visited at most once. 

\subsection{Soundness of Strategy-based Verification}

The game $\game{\tss}{\varphi}$ mimics the strategic behavior of the $\exists$-player.
In each step, the refuter chooses successors for the $k$ universally quantified traces, followed by the verifier who selects successors for the $l$ existentially quantified traces.
The automaton state in the nodes of $\game{\tss}{\varphi}$ tracks the (unique) run of $\aut^\phi$ on the resulting $k+l$ traces. 
To verify that $\tss \models \varphi$, the verifier should win from every possible combination of initial states for the universally quantified copies. 
We define
\begin{align*}
    V_{\mathit{init}} \coloneqq \{(s_1, \ldots, s_k) \mid \forall 1 \leq i \leq k\ldot s_i \in S_0\}.
\end{align*}
We write $\winss{\veri}{\game{\tss}{\varphi}}$ if the verifier wins $\game{\tss}{\varphi}$ from all nodes in $V_{\mathit{init}}$.
We can show the soundness of our verification method.

\begin{restatable}{theorem}{soundnessTheo}\label{theo:soundness}
    If $\winss{\veri}{\game{\tss}{\varphi}}$ then $\tss \models \varphi$.
\end{restatable}
\begin{proofSketch}
    We use a positional winning strategy $\sigma$ for $\veri$ that witnesses $\winss{\veri}{\game{\tss}{\varphi}}$ to iteratively construct traces for the existentially quantified traces by simulating $\sigma$ on finite prefixes of the universally quantified traces.
    We give a detailed proof in \ifFull{\refApp{sec:appSoundness}}{the full version \cite{fullVersion}}.
\end{proofSketch}

\subsection{Prophecies and Prophecy Variables}

As we saw in \Cref{sec:runningExample}, strategy-based verification of $\forall^*\exists^*$ properties is incomplete, i.e., 
$\veri$ might lose $\game{\tss}{\varphi}$ even though the system satisfies the property.
Intuitively, this is the case when the $\exists$-player (the verifier in $\game{\tss}{\varphi}$) needs future information that is not available by observing only a prefix of the universally quantified traces.
To counteract this lack of information, we introduce prophecies.

\begin{definition}
    A \emph{prophecy} is a $\omega$-regular subset $\proph \subseteq (\Sigma^k)^\omega$.  
\end{definition}

If a prophecy $\proph$ holds at step $i$, the $\exists$-player can assume that the $\forall$-player (the refuter in $\game{\tss}{\varphi}$) starting in step $i$, constructs traces $t_1, \ldots, t_k$ for the $k$ universal quantifiers such that $\zip{t_1, \ldots, t_k} \in \proph$.
The prophecy thereby provides limited information (in form of the binary information on whether or not the prophecy holds) about the future behavior of the universally quantified traces. 

To formally introduce prophecies into our framework, we need to enable the $\exists$-player to, in each step, determine which prophecies hold.  
We delegate this step to the universal player who determines the truth value for each prophecy in its (modified) state space.
Formally, we accomplish this in two steps.
(1) We extend the system by fresh boolean variables (called \emph{prophecy variables}) that, in each step, can be chosen non-deterministically, and (2) we relax the specification to ensure that the prophecy variables set by the $\forall$-player correspond to the truth value of the prophecies.

\subsubsection{System Manipulation}
We begin by modifying the transition system to allow the $\forall$-player to set the prophecy variables.

\begin{definition}\label{def:addProphToTs}
    Given a transition system $\tss = (S, S_0, \varrho, L)$ and a set of fresh propositions $P$ (with $P \cap \ap = \emptyset$) we define the modified transition system $\tss^P \coloneqq (S^P, S_0^P, \varrho^P, L^P)$ over $\ap \mathbin{\mathaccent\cdot\cup} P$ where $S^P \coloneqq S \times 2^P$, $S_0^P \coloneqq S_0 \times 2^P$, $\varrho^P \coloneqq \{ ((s, A), (s', A')) \mid (s, s') \in \varrho \, \land  \,A, A' \in 2^P \}$ and $L^P(s, A) \coloneqq L(s) \cup A$.
\end{definition}

In particular, we have 
\begin{align*}
    \mathit{Traces}(\tss^P) = \{ t \cup t' \mid t \in \mathit{Traces}(\tss), t' \in (2^{P})^\omega \} 
\end{align*}
where $t \cup t'$ denotes the pointwise union of both traces. 

\subsubsection{Property Manipulation}
We modify the matrix of the hyperproperty such that the original property is only required to hold, if all prophecies by the universal player are set correctly, i.e., a prophecy variable in $P$ is set to true iff the universally quantified traces produced by the $\forall$-player are contained in the corresponding prophecy. 
To express this at the logical level, we make use of the fact that we can express a ($\omega$-regular) prophecy $\proph$ as a QPTL formula (cf.~\Cref{theo:QPTL}).\footnote{In practice, we would not express prophecies in QPTL and instead operate directly on an automaton-based representation of a prophecy.
    By taking this detour, we can keep the notation succinct and can express the assumption that the $\forall$-player correctly sets the prophecy variables as a logical implication.}

\begin{definition}\label{def:addProphToForm}
    Given a set of QPTL formulas $\Xi = \{\xi_1, \ldots, \xi_n\}$ using only trace variables in $\{\pi_1, \ldots, \pi_k\}$  and a fresh set of atomic propositions $P = \{p_1, \ldots, p_n\}$, define the modified formula $\varphi^{P, \Xi}$ as 
    \begin{align*}
        \forall \pi_1 &\ldots \forall \pi_k. \exists \pi_{k+1} \ldots \exists \pi_{k+l}\ldot  \bigg[\ltlG \bigwedge_{j=1}^n ({p_j}_{\pi_1} \leftrightarrow {\xi_j})\bigg]\to \phi.
    \end{align*}
\end{definition}

That is, we only require $\phi$ to hold, if in every step and for every $1 \leq j \leq n$, the prophecy formula $\xi_j$ holds exactly when the prophecy variable $p_j$ is set on trace $\pi_1$.\footnote{
With the construction of $\varphi^{P, \Xi}$ we ensure that each prophecy variable on $\pi_1$ reflects the truth value of the prophecy.
However, any of the universally quantified trace variables would work equally well. 
Note that each prophecy formula $\xi_j$ captures a behavior of the combined executions of the universally quantified traces $\pi_1, \ldots, \pi_k$ (as $\xi_j$ uses trace variables in $\{\pi_1, \ldots, \pi_k\}$) and not necessarily the behavior of a single trace.
In fact, local prophecies (i.e., prophecies that only capture behavior on one trace) are insufficient for completeness (cf.~\Cref{ex:multipleUniversalTraces}).
}

\subsubsection{Soundness of Prophecies}

The combination of the modified transition system (which allows the prophecy variables to take any value) and the modified property does not impact the satisfaction of the original property on the original system as stated in the following theorem (see, e.g., \cite[Thm.~5]{CoenenFST19}).

\begin{theorem}\label{theo:introProph}
    Let $\Xi = \{\xi_1, \ldots, \xi_n\}$ and $P = \{p_1, \ldots, p_n\}$ be as in \Cref{def:addProphToForm}. 
    Then $\tss \models \varphi$ if and only if $\tss^P \models \varphi^{P, \Xi}$.
\end{theorem}

\begin{remark}
    A brief remark about nomenclature is in order.
    A \emph{prophecy} is a $\omega$-regular set of traces $\proph$. 
    We represent this prophecy as a QPTL formula $\xi_i \in \Xi$ which we also refer to as a prophecy or \emph{prophecy formula}.
    Lastly, $p_i \in P$ is a \emph{prophecy variable} that corresponds to prophecy (formula) $\xi_i$.
\end{remark}

\begin{figure*}[!t]
	
	\begin{subfigure}{1\columnwidth}
		\centering
		\begin{tikzpicture}[scale=1]
			\node[] at (0,0) (n0) {$\winss{\veri}{\game{\tss^P}{\varphi^{P, \Xi}}}$};
			\node[] at (-3,-1.5) (n1) {$\tss \models \varphi$};
			
			\node[] at (0,-1.5) (n3) {$\tss^P \models \varphi^{P, \Xi}$};

			\draw[->, thick] (n0) --node[right] {\scriptsize Thm.~\ref*{theo:soundness}} (n3);
			
			\draw[->, thick,transform canvas={yshift=0.5ex}] (n1) -- (n3);
			\draw[->, thick,transform canvas={yshift=-0.5ex}] (n3) --node[below] {\scriptsize Thm.~\ref*{theo:introProph}} (n1);
			
		\end{tikzpicture}
		\subcaption{Implications for any set of prophecies.}\label{fig:results}
	\end{subfigure}%
	\begin{subfigure}{1\columnwidth}
		\centering
		\begin{tikzpicture}[scale=1]
			\node[] at (0,0) (n0) {$\winss{\veri}{\game{\tss^{{P}}}{\varphi^{{P}, {\Xi}}}}$};
			\node[] at (-4,-1.5) (n1) {$\tss \models \varphi$};

			\node[] at (0,-1.5) (n3) {$\tss^P \models \varphi^{{P}, {\Xi}}$};
			
			\draw[->, thick] (n0) --node[right] {\scriptsize Thm.~\ref*{theo:soundness}} (n3);
			
			\draw[->, thick,transform canvas={yshift=0.5ex}] (n1) -- (n3);
			\draw[->, thick,transform canvas={yshift=-0.5ex}] (n3) --node[below] {\scriptsize Thm.~\ref*{theo:introProph}} (n1);

			\draw[->, thick] (n1) --node[above, sloped] {\scriptsize Thm.~\ref*{theo:comp}} (n0);
		\end{tikzpicture}
		\subcaption{Implications for a \emph{complete} set of prophecies.}\label{fig:resultsComp}
	\end{subfigure}

	\caption{Implications between the satisfaction of a hyperproperty and the existence of a winning strategy for the $\exists$-player.
		We display implications with an arbitrary set of prophecies (\Cref{fig:results}) and a complete set (\Cref{fig:resultsComp}). 
	}
\end{figure*}
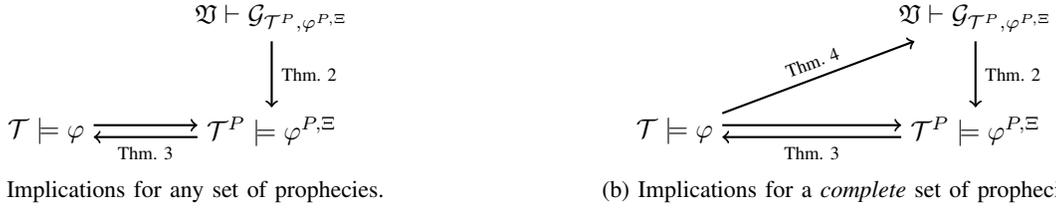

\subsubsection{Prophecies for Strategy-based Verification}

While the addition of prophecies does not alter the satisfaction of the property in the HyperLTL semantics (as stated in \Cref{theo:introProph}), it can impact the existence of a winning strategy for the $\exists$-player during strategy-based verification. 
 That is, it might be that $\winss{\veri}{\game{\tss}{\varphi}}$ does not hold, but $\winss{\veri}{\game{\tss^P}{\varphi^{P, \Xi}}}$ does.
 Thus, prophecies provide a natural tool to strengthen strategy-based verification and allow the user to, e.g., introduce domain knowledge in the form of user-defined prophecies. 
 The soundness of the addition of prophecies can be argued easily: 
If $\winss{\veri}{\game{\tss^P}{\varphi^{P, \Xi}}}$ holds, then (by \Cref{theo:soundness}) $\tss^P \models \varphi^{P, \Xi}$ so (by \Cref{theo:introProph}) $\tss \models \varphi$.
 The situation is depicted graphically in \Cref{fig:results}. 

\begin{example}\label{ex:simpleExample}
	With our notation fixed, we revisit the transition system $\tss$ and HyperLTL formula $\varphi$ from \Cref{sec:runningExample}.
	In this case, $\veri \not\vdash \game{\tss}{\varphi}$, i.e., strategy-based verification fails without the addition of prophecies.
	Let $\Xi = \{\ltlN a_\pi\}$ and let $P = \{p\}$ be a fresh set of prophecies variables. Using \Cref{def:addProphToForm} we construct
	\begin{align*}
		\varphi^{P, \Xi} = \forall \pi. \exists \pi'\ldot \ltlG (p_\pi \leftrightarrow \ltlN a_\pi) \to \ltlG (a_{\pi'} \leftrightarrow \ltlN a_\pi).
	\end{align*}
	It is easy to see that $\winss{\veri}{\game{\tss^P}{\varphi^{P, \Xi}}}$:
	the prophecy variable $p$ hints at the next move of $\refu$. 
	If, for example, $\refu$ sets $p$ to true, $\veri$ can assume that $\ltlN a_\pi$ holds (if it does not, the premise of $\varphi^{P, \Xi}$ is violated and so the play is trivially won by $\veri$).
	The verifier can thus move to state $s_1$ (in \Cref{fig:ilEx1}) and thereby correctly predict the next move on $\pi$.
	As argued in \Cref{fig:results}, $\winss{\veri}{\game{\tss^P}{\varphi^{P, \Xi}}}$ implies that $\tss \models \varphi$.
\end{example}

\section{Completeness}\label{sec:comp}

We have argued that strategy-based verification remains sound when adding prophecies.
The natural question that arises is the following: 
\begin{center}
	\parbox{0.85\columnwidth}{
		\begin{center}
			\textit{Assume that $\tss \models \varphi$. Does there exist \emph{some} finite set of prophecies ${\Xi}$ such that $\winss{\veri}{\game{\tss^{{P}}}{\varphi^{{P}, {\Xi}}}}$?}
		\end{center}
	}
\end{center}%
As already observed by Coenen et al.~\cite{CoenenFST19}, this does not hold if we only allow LTL-definable prophecies. 

\begin{example}\label{ex:ltlex}
    Consider a system $\tss$ that generates all traces over $\ap = \{a, b\}$ and the following property $\varphi$
    \begin{align*}
      \forall \pi.\exists \pi'\ldot a_{\pi'} \land \ltlG (a_{\pi'} \leftrightarrow \ltlN \neg a_{\pi'}) \land (b_{\pi'} \leftrightarrow \ltlF (b_{\pi} \land a_{\pi'})).
    \end{align*}
	That is, $b$ should hold in the first step on $\pi'$ iff $b$ holds at some \emph{even} position on $\pi$.
    Clearly, $\tss \models \varphi$ but $\veri \not\vdash \game{\tss}{\varphi}$.
    Moreover, LTL cannot express that (on $\pi$) $b$ ever holds at an even position (cf.~\Cref{ex:qptlProperty}).
    Consequently, no LTL-definable prophecy can provide sufficient information to the $\exists$-player, i.e., $\veri \not\vdash \game{\tss^{{P}}}{\varphi^{{P}, {\Xi}}}$ for any (finite) set of LTL formulas $\Xi$.
\end{example}

While this incompleteness result for LTL-definable prophecies is interesting in its own right, we usually do not represent prophecies as LTL formulas but work with some automaton representation. 
Consequently, we are less interested in LTL-definable prophecies but in the existence of $\omega$-regular prophecies. 
And indeed, in this paper, we show that we can answer the above question positively if we shift from LTL-definable prophecies to arbitrary $\omega$-regular prophecies.
The main result of this paper reads as follows:

\begin{theorem}\label{theo:comp}
	Let $\tss$ be a (finite-state) transition system and let $\varphi$ be a $\forall^*\exists^*$ HyperLTL property such that $\tss \models \varphi$.
	There exist finitely many QPTL prophecies ${\Xi} = \{\xi_1, \ldots, \xi_n\}$ such that for a fresh set ${P} = \{p_1, \ldots, p_n\}$ we get $\winss{\veri}{\game{\tss^{{P}}}{\varphi^{{P}, {\Xi}}}}$.
\end{theorem}

If $\tss \models \varphi$ we call a set of prophecies ${\Xi}$ \emph{complete} if $\winss{\veri}{\game{\tss^{{P}}}{\varphi^{{P}, {\Xi}}}}$, i.e., $\Xi$ is a witness to \Cref{theo:comp}.
The resulting situation is depicted in \Cref{fig:resultsComp}. 
Combined with \Cref{theo:soundness} and \Cref{theo:introProph} we can rephrase \Cref{theo:comp} as follows:

\begin{corollary}
		Let $\tss$ be a (finite-state) transition system and let $\varphi$ be a $\forall^*\exists^*$ HyperLTL property.
		There exist finitely many QPTL prophecies ${\Xi} = \{\xi_1, \ldots, \xi_n\}$ such that for a fresh set ${P} = \{p_1, \ldots, p_n\}$ we get $\winss{\veri}{\game{\tss^{{P}}}{\varphi^{{P}, {\Xi}}}}$ if and only if $\tss \models \varphi$.
\end{corollary}

We note that our prophecy construction used to prove \Cref{theo:comp} yields prophecies \emph{without} first checking if $\tss \models \varphi$.
This allows us to use our construction to (algorithmically) check if $\tss \models \varphi$ (we discuss this in \Cref{sec:mc}).

\begin{remark}
	We can strengthen \Cref{theo:comp} further.
	Our prophecy construction treats the LTL matrix of the HyperLTL property as an $\omega$-automaton.
	The constructions thus generalize to all logics that utilize the trace quantification mechanism of HyperLTL but express arbitrary $\omega$-regular property within their matrix.
	For example, our result also applies to HyperQPTL, i.e., formulas where the trace-quantifier prefix is followed by a QPTL formula. 
	We thus show that $\omega$-regular prophecies suffice for all $\forall^*\exists^*$ hyperproperties with $\omega$-regular matrix.
	In contrast,  \Cref{ex:ltlex} shows that LTL-definable prophecies are not sufficient for $\forall^*\exists^*$ hyperproperties with LTL-definable matrix (aka.~HyperLTL). 
\end{remark}

\begin{example}\label{ex:multipleUniversalTraces}
	We can show that in the case of more than a single universally quantified trace (i.e., cases where $k > 1$), prophecies must necessarily reason about the joint future behavior of all $k$ universally quantified traces.
	Consider the transition system $\tss$ in \Cref{fig:ilEx1} that generates all traces over $\ap = \{a\}$ and the property 
	\begin{align*}
		\varphi = \forall \pi. \forall \pi'. \exists \pi''\ldot a_{\pi''} \leftrightarrow \ltlG (a_{\pi} \leftrightarrow a_{\pi'}).
	\end{align*}
	That is, $a$ should hold on $\pi''$ in the first step iff $\pi$ and $\pi'$ are equal.
	Clearly, $\tss \models \varphi$ but this cannot be verified using strategy-based verification without prophecies.
	The LTL-definable prophecy $\xi \coloneqq \ltlG (a_{\pi} \leftrightarrow a_{\pi'})$ provides enough information to the $\exists$-player on whether or not to set $a$ in the first step.
	However, any finite set of local prophecies (i.e., prophecy formulas that only refer to $\pi$ or only refer to $\pi'$) is incomplete.
\end{example}

The following two sections are devoted to a proof of \Cref{theo:comp}.
To avoid clustered notation, we give our proof for hyperproperties of the form $\forall \pi. \exists \pi'. \phi$.
Our result generalizes easily to the entire $\forall^*\exists^*$ fragment.
We begin our proof by considering HyperLTL formulas of the form $\forall \pi. \exists \pi'. \phi$ where $\phi$, when interpreted as a trace property, is a \emph{safety} property (in the traditional sense \cite{AlpernS85}).
This allows for a simpler construction (in \Cref{sec:compSafety}).
In \Cref{sec:compGeneral} we then incrementally extend the construction to general temporal properties.

\begin{remark}
	It is important to note that the class of safety used in \Cref{sec:compSafety} only refers to the LTL matrix (the body) of the HyperLTL property.
	If the matrix is safety, this does \emph{not} imply that the HyperLTL formula is hypersafety (as defined by Clarkson and Schneider \cite{ClarksonS08}).
	For example, the matrix of GNI (cf.~\Cref{sec:intro}) is a safety property (and thus lends itself to the simpler construction in \Cref{sec:compSafety}), but GNI is a hyperliveness property \cite{ClarksonS08,CoenenFST19}.
	On the other hand, as shown in \cite{BeutnerCFHK22}, the class of formulas with safety matrix (called \emph{temporal safety} in \cite{BeutnerCFHK22}) already contains all $\forall^*\exists^*$ hypersafety properties.
\end{remark}

\section{Completeness for Safety Matrix}\label{sec:compSafety}

We first consider the case where $\phi$ is a safety property.
Let $\aut^\phi = (Q^\phi, q^\phi_0, \delta^\phi, B^\phi)$ be a deterministic safety automaton over $\Sigma \times \Sigma$ for $\phi$.

\subsection{Prophecy Construction}\label{sec:safetyConstruction}

The main idea behind our completeness result (which in a modified form also applies to the general case in \Cref{sec:compGeneral}) is to design prophecies that directly identify those states that the $\exists$-player should move to. 
As we assume that $\phi$ denotes a safety property, we can accomplish this by identifying all states that are \emph{safe}, i.e., all states that the $\exists$-player can move to without losing the game immediately. 
Formally, we add a prophecy for each state $s$ of the game and design them such that a trace constructed by the $\forall$-player lies within a prophecy for state $s$ if and only if choosing $s$ as a successor is safe for the $\exists$-player.
For every $q \in Q^\phi$  and $s \in S$ we define
\begin{align*}
	\lpro{q}{s}\coloneqq \{ t \in \Sigma^\omega \mid \exists t' \in \mathit{Traces}(\tss_s)\ldot  t \otimes t' \in \calL(\aut^\phi_q)  \}.
\end{align*}
Recall that $\tss_s$ is $\tss$ with $s$ fixed as the initial state and similarly for $\aut^\phi_q$. 
That is, a trace $t$ (chosen for the universally quantified trace in $\varphi$)  is in $\lpro{q}{s}$ if there exists some trace (chosen for the existentially quantified trace in $\varphi$) that starts in $s$ and, in combination with $t$, is accepted by $\aut^\phi$ (when starting in $q$).

To (informally) see why these prophecies are useful for the $\exists$-player, let us assume that the current state of $\aut^\phi$ (on the current prefix of the game) is $q$.
If prophecy $\lpro{q}{s}$ holds, the $\exists$-player can move to state $s$ knowing that the $\forall$-player plays such that $s$ is a safe move (as \emph{some} trace starting from $s$ is still winning).

\begin{remark}
	In our formalization, the $\forall$-player sets the prophecy variables.
	Conceptually, we can thus consider prophecies as a binding contract between the $\forall$-player and the $\exists$-player. 
	When the $\forall$-player indicates that $\lpro{q}{s}$ holds (by setting the respective prophecy variable), the $\forall$-player enters a binding agreement that guarantees that the constructed trace is contained in $\lpro{q}{s}$ (as otherwise, the premise of $\varphi^{P, \Xi}$ is violated so the $\exists$-player wins trivially).
	From this point of view, our prophecies defer the selection of a successor state from the $\exists$-player to the $\forall$-player: By setting the variables, the $\forall$-player implicitly fixes all valid moves for the $\exists$-player. 
\end{remark}

We can easily see that the resulting prophecies are $\omega$-regular (by constructing the product of $\tss_s$ and $\aut^\phi_q$). 
Consequently, we can represent  each prophecy $\lpro{q}{s}$ as a QPTL prophecy formula $\pro{q}{s}$ (cf.~\Cref{theo:QPTL}).
The resulting set of prophecies is complete in the sense of \Cref{theo:comp}.

\begin{restatable}{theorem}{compltnessSafety}\label{theo:compSafety}
	Assume $\tss \models \varphi$.
	Define $\Xi = \{ \pro{q}{s} \mid q \in Q^\phi, s \in S \}$ and let $P = \{\prov{q}{s} \mid q \in Q^\phi, s \in S\}$ be  a  fresh set of atomic propositions. 
    Then $\winss{\veri}{\game{\tss^P}{\varphi^{P, \Xi}}}$. 
\end{restatable}

\subsection{Correctness Proof}\label{sec:safetyProof}

In this subsection, we sketch a proof of \Cref{theo:compSafety}.
As a complete proof is rather involved, we restrict ourselves to the construction of a winning strategy for the $\exists$-player and refer to a detailed proof in \ifFull{\refApp{sec:appSafetyCase}}{the full version \cite{fullVersion}}.
Readers less interested in the proof can skip to \Cref{sec:prophCount}.

\subsubsection{Notation}\label{sec:safetyNotation}

We begin by introducing some notation. 
By definition of $\tss^P$, nodes in $\game{\tss^P}{\varphi^{P, \Xi}}$ either have the form $(s, A)$, where $s \in S$ and $A \subseteq P$ or the form $\gamenode{(s, A), (s', A'), q, \flat}$, where $s, s' \in S$, $A, A' \subseteq P$ and $\flat \in \{\forall, \exists\}$. 
Here $q$ is an automaton state in a DPA tracking
\begin{align}\label{eq:LTL}
    \phi^{P, \Xi} \coloneqq \ltlG \Big(\!\!\! \bigwedge_{q \in Q^\phi, s \in S} \!\!\! (\prov{q}{s})_{\pi}  \leftrightarrow \pro{q}{s}\Big)  \rightarrow \phi.
\end{align}%
It is easy to see that in states of the form $\gamenode{(s, A), (s', A'), q, \flat}$ the $A'$ component (stemming from the definition of $\tss^P$) is irrelevant as in (\ref{eq:LTL}) the prophecy variables are only referred to on trace variable $\pi$. 
We, therefore, consider a node $\gamenode{(s, A), (s', A'), q, \flat}$ simply as $\gamenode{(s, A), s', q, \flat}$.
With this conceptual simplification, any finite play $\vertex  \in V^* $ in $\game{\tss^P}{\varphi^{P, \Xi}}$ (starting in some state in $V_\mathit{init}$) of odd-length (where $|\vertex| = 2i + 1$) has the form 
\begin{align}\label{eq:examplePath}
    \begin{split}
        &(s_0, A_0) \to \gamenode{(s_0, A_0), s_0', q_0, \forall} \to \gamenode{(s_1, A_1), s_0', q_1, \exists}\\
        &\to \gamenode{(s_1, A_1), s_1', q_1, \forall} \to \cdots \to \gamenode{(s_i, A_i), s_{i-1}', q_i, \exists}.
    \end{split}
\end{align}
We can extract from $\vertex$ both paths through $\tss$ and the prophecy variables set at each step.
Define $\sSeq{0}\sSeq{1} \cdots \sSeq{i}$ to be the path of the $\forall$-player ($s_0s_1\cdots s_i$ in (\ref{eq:examplePath})), $\bSeq{0}\bSeq{1}\cdots\bSeq{i}$ the sequence of prophecy variables chosen ($A_0A_1 \cdots A_i$ in (\ref{eq:examplePath})), and $\ssSeq{0}\ssSeq{1}\cdots \ssSeq{i-1}$ the path for the $\exists$-player ($s_0's_1'\cdots s_{i-1}'$ in (\ref{eq:examplePath})).
Define $\sTr{k} \coloneqq L(\sSeq{k})$ and $\ssTr{k} \coloneqq L(\ssSeq{k})$ for $0 \leq k \leq i-1$.

\subsubsection{Strategy Construction}

With those definitions at hand, we define an explicit winning strategy $\sigma$ for $\veri$ as follows:

\begin{algorithm}[H]
	\begin{algorithmic}[1]
		\State \textbf{Input:} $\vertex \in V^*$ \textbf{with} $|\vertex| = 2i + 1$
		\If{$i = 0$}
		\State  $T \coloneqq S_0$     \label{line:l1}
		\Else
		\State  $T \coloneqq \sucs{\ssSeq{i-1}}$\label{line:l2}\vspace{-0.8mm}
		\EndIf
		\State $\hat{q} \coloneqq {\delta^{\phi}}^* \big[(\sTr{0}, \ssTr{0}) \cdots (\sTr{i-1},\ssTr{i-1})\big]$\label{line:l3}
		\State $C \coloneqq \{ s' \mid s' \in T \land  \prov{\hat{q}}{s'} \in \bSeq{i} \}$\label{line:l4}
		\If{$C \neq \emptyset$}\label{line:l5}
		\State \textbf{return} any $s' \in C$\label{line:l6}
		\Else
		\State \textbf{return} any $s' \in T$\label{line:l7}
		\EndIf
	\end{algorithmic}
\end{algorithm}

\vspace{-5mm}

\noindent
Note that $\sigma$ directly returns a successor state in $\tss$.

By the structure of $\game{\tss^P}{\varphi^{P, \Xi}}$, any finite path starting in $V_{\mathit{init}}$ that reaches a node in $V_\veri$ is of odd length.
We begin by computing all possible successor states for the $\exists$-player in a set $T$. These are either all initial nodes in the case where $|\vertex| = 1$ (line \ref*{line:l1}) or all successor states of the current state of the $\exists$-player (line \ref*{line:l2}).
We then compute the state $\hat{q}$ of $\aut^\phi$ reached on $\vertex$ in line \ref*{line:l3}.
Note that $\hat{q}$ is a state in $\aut^\phi$ whereas the automaton states occurring in $\vertex$ are states in a DPA tracking (\ref{eq:LTL}).
In line \ref*{line:l4}, we check if any of the possible successors in $T$ are declared safe by the $\forall$-player, i.e., we check for states where the corresponding prophecy variable is set. 
If there is any such state, we pick it (line \ref*{line:l6}).
Otherwise, we choose an arbitrary successor (line \ref*{line:l7}).

\begin{example}
    We can simulate the strategy on abstract prefixes of (\ref{eq:examplePath}).
    Initially, for $\vertex = (s_0, A_0)$ it picks any initial state $s_0' \in S_0$ such that  $\prov{s_0'}{q^\phi_0} \in A_0$.
    For path $\vertex = (s_0, A_0) \to \gamenode{(s_0, A_0), s_0', q_0, \forall} \to \gamenode{(s_1, A_1), s_0', q_1, \exists}$ it computes the current state $\hat{q}$ of $\aut^\phi$ reached on the path $(L(s_0), L(s_0')) \in (\Sigma \times \Sigma)^*$ and picks any successor $s_1'$ of $s_0'$ such that $\prov{s_1'}{\hat{q}} \in A_1$.
\end{example}

It remains to argue the correctness of the just constructed strategy. 
Here, we may assume that all prophecies are set correctly (i.e., the premise of (\ref{eq:LTL}) is true) as otherwise, the play is trivially won by $\veri$. 
Under this assumption, the premise that $\tss \models \varphi$, and by induction on the length of a prefix of (\ref{eq:examplePath}) we can establish that $C$ (as computed in line \ref*{line:l4}) is never empty, so the strategy always selects a successor for which the prophecy holds.
This already implies that the play is winning for the $\exists$-player:
Indeed, if any state $\hat{q}$ in $\aut^\phi$ were bad, we would get $\lpro{\hat{q}}{s} = \emptyset$ for all states $s$, and so the set $C$ computed in line \ref*{line:l4} would be empty as well (as we assumed that the prophecy variables are set correctly).
A detailed proof can be found in \ifFull{\refApp{sec:appSafetyCase}}{the full version \cite{fullVersion}}.

\subsection{On the Number Of Prophecies}\label{sec:prophCount}

As established in \Cref{theo:compSafety}, the size of a complete set of prophecies  is upper bounded by $|S| \cdot |Q^\phi|$.
We can restrict the number of prophecies further (which is relevant in practice but does not offer an asymptotic improvement).
Two states $s_1, s_2$ are trace equivalent, written $s_1 \equiv_\mathit{Trace} s_2$, if $\mathit{Traces}(\tss_{s_1}) = \mathit{Traces}(\tss_{s_2})$.
If $s_1 \equiv_\mathit{Trace} s_2$, we get $\lpro{q}{s_1} = \lpro{q}{s_2}$ for any automaton state $q$, so we can restrict the prophecy construction to the equivalence class of $\equiv_\mathit{Trace}$.

We do not claim that our explicit prophecy construction in \Cref{sec:safetyConstruction} is optimal w.r.t.~the number of prophecies.
We can, however, show that the number of prophecies must necessarily grow with the size of the system, i.e., it cannot be constant (see \ifFull{\refApp{app:lowerBound}}{the full version \cite{fullVersion}} for a proof).

\begin{restatable}{proposition}{lognprop}\label{prop:lb}
    There exists a $\forall\exists$ HyperLTL property $\varphi$ with safety matrix and a family of transition systems $\{\tss_n\}_{n \in \nat}$ such that $\tss_n$ has $\Theta(n)$-many states, and $\tss_n \models \varphi$, and, additionally, any family of prophecies $\{\Xi_n\}_{n \in \nat}$ where $\Xi_n$ is complete for $\tss_n, \varphi$ has at least size $|\Xi_n| \in \Omega(\log n)$.
\end{restatable}

\section{Completeness for $\omega$-regular Matrix}\label{sec:compGeneral}

So far, the prophecy construction from \Cref{sec:compSafety} is limited to the case where $\phi$ is a safety property. 
In this section, we incrementally modify the construction to support properties where $\phi$ expresses arbitrary $\omega$-regular properties.
To begin with, it is helpful to analyze \emph{why} the construction from \Cref{sec:compSafety} fails when moving beyond safety.

\begin{example}\label{ex:inc}
    As a simple example to see this, we again consider the transition system in \Cref{fig:ilEx1}, generating all traces over $\ap = \{a\}$.
    Define 
    \begin{align*}
        \varphi \coloneqq \forall \pi. \exists \pi'\ldot \ltlG \ltlF\big( a_{\pi'} \leftrightarrow \ltlN a_\pi\big)
    \end{align*}
    which expresses that $\pi'$ should predict the \emph{next} step on $\pi$. 
    Importantly, $\pi'$ should not necessarily predict the next step of $\pi$ at every point but at least infinitely many times. 
    Clearly, $\tss \models \varphi$ but $\veri \not\vdash \game{\tss}{\varphi}$.
    Let $\Xi$ be the set of prophecies constructed in \Cref{sec:compSafety}.
    It is easy to see that for any reachable state $q$ (in the canonical DPA for the matrix of $\varphi$) and any trace $t$, it holds that $t \in \lpro{q}{s_1}$ and $t \in \lpro{q}{s_2}$.
    This is the case as choosing either state is safe, i.e., does lose the game for the $\exists$-player. Even if the current prediction is incorrect, infinitely many correct predictions are still possible in the future. 
    The $\forall$-player can therefore set all prophecy variables to true without invalidating the premise of $\varphi^{P, \Xi}$, so $\veri \not\vdash \game{\tss^P}{\varphi^{P, \Xi}}$; $\Xi$ is an incomplete set of prophecies. 
\end{example}

As evident in \Cref{ex:inc}, the root cause is that the safety prophecies provide enough information to never lose the game (i.e., encounter a situation from which the game cannot be won anymore), but this does (when moving beyond safety) not guarantee that the game is won.

\subsection{Optimal Successors and Prophecy Construction}\label{sec:constructionBuechi}

We begin our extension to support full $\omega$-regularity  by assuming that we can express $\phi$ with a \emph{deterministic} Büchi automaton $\aut^\phi = (Q^\phi, q^\phi_0, \delta^\phi, F^\phi)$.
This is a proper extension of the safety case in \Cref{sec:compSafety} (as all safety properties can be expressed with a deterministic Büchi automaton) but does not capture full $\omega$-regularity yet (we relax this further in \Cref{sec:beyond}).
Note that the property in \Cref{ex:inc} can be recognized by a deterministic Büchi automaton.
We further assume, w.l.o.g., that $\tss$ has a unique initial state.

Following the idea from \Cref{sec:compSafety}, the prophecies should explicitly tell the $\exists$-player which successor to choose. 
The crucial idea underlying our construction is that the prophecies should point to successor states that are safe and, additionally, satisfy that the next visit to an accepting state in $F^\phi$ occurs \emph{as fast as possible} (where the speed refers to the number of steps).
Always choosing such an ``optimal'' successor guarantees that an accepting state is visited infinitely many times. 
A na\"ive idea where prophecies express ``state $s$ is safe, and a visit to an accepting state is possible in $n$ step'' would certainly work (a strategy for the $\exists$-player would always pick a successor where the number of steps is minimal) but cannot be expressed in finitely many prophecies (the number of steps needs to be unbounded). 
The core idea in this section is to express optimality of a state by means of a \emph{relative} compression with possible alternative states.
Perhaps surprisingly, this is possible within the framework of $\omega$-regular prophecies.
For automaton state $q \in Q^\phi$ and system states $x, s \in S$ with $s \in \sucs{x}$ define $\lproRel{q}{x}{s}$ as follows:
 \begin{align*}
    \bigg\{ t &\in \Sigma^\omega \mid \exists t' \in \mathit{Traces}(\tss_{s})\ldot t \otimes t' \in \calL(\aut^\phi_q) \,\land \\
    &\Big[\forall s' \in \sucs{x}\ldot \forall t'' \in \mathit{Traces}(\tss_{s'})\ldot t \otimes t'' \in \calL(\aut^\phi_q)\\
    &\Rightarrow \firstvisit{F^\phi}{\aut^\phi_q}{t \otimes t'}  \leq  \firstvisit{F^\phi}{\aut^\phi_q}{t \otimes t''}\Big] \bigg\}
\end{align*}
Recall that $ \mathit{firstVisit}_{F^\phi}(\aut^\phi_q, t \otimes t')$ denotes the first time point that the unique run on $t \otimes t'$ visits $F^\phi$.
The first line in our new definition is similar to the safety case, i.e., a trace $t$ is in $\lproRel{q}{x}{s}$ if there exists a witness trace $t'$ starting in $s$.
In addition, we require that for any alternative successor $s'$ of $x$ and all traces $t''$ starting in $s'$ that are also winning, the first visit to an accepting state in $F^\phi$ occurs at least as fast on $t'$ as on the alternative trace $t''$.
For a given trace $t \in \lproRel{q}{x}{s}$, choosing $s$ as the successor of $x$ is thus optimal, in the sense that from no other successor of $x$ there is a witness trace that visits an accepting state (strictly) sooner. 

\begin{example}
    We revisit \Cref{ex:inc}.
    With our new construction, a trace $t$ satisfies $t \in \lproRel{q}{x}{s_1}$ (for any automaton state $q$ and any $x \in \{s_1, s_2\}$) if and only if $a \in t(1)$.
    That is, choosing $s_1$ as a successor is optimal iff this correctly predicts the next state of the $\forall$-player, i.e., $a$ holds in the next step on $t$.
    In particular, correctly predicting the move of the $\forall$-player in the current step is better (measured in the number of steps to an accepting state) than misspredicting it now but predicting it correctly sometime in the future. 
    A strategy that follows the recommendations of the new prophecies will \emph{always} (instead of only infinitely many times) correctly predict the next step of the $\forall$-player and is therefore winning. 
\end{example}

\subsection{On $\omega$-Regularity}
It is not immediate that $\lproRel{q}{x}{s}$ is $\omega$-regular.
We begin by showing the following.

\begin{proposition}\label{prop:omegaReagular}
    For any $q \in Q^\phi$ and $s \in \sucs{x}$, $\lproRel{q}{x}{s}$ is $\omega$-regular.
\end{proposition}
\begin{proof}
    We show that we can express $\lproRel{q}{x}{s}$ as a QPTL formula which already gives the desired $\omega$-regularity by \Cref{theo:QPTL}.
    We make heavy use of propositional quantification in QPTL to encode paths in $\tss$ and corresponding runs of $\aut^\phi$.
    Let $V_\tss \coloneqq \{p_s \mid s \in S\}$ and $V_{\aut^\phi} \coloneqq \{p_q \mid q \in Q^\phi\}$  be propositional QPTL variables.
    We define formula $\mathit{valid}$ as\\
 \scalebox{0.95}{
    \parbox{1\linewidth}{
        \begin{align*}
            \ltlG \bigg[ \mathit{un}(V_\tss) \land \mathit{un}(V_{\aut^\phi})\land \Big( \!\!\!\!\!\!\!\!\!\!\!\!\! \bigvee_{\substack{s \xrightarrow{\tss} s' \\ \delta^\phi(q, (\sigma, L(s))) = q'}} \!\!\!\!\!\!\!\!\!\!\!\!\!  p_s \land \LTLnext p_{s'} \land p_q \land  \LTLnext p_{q'} \land \sigma  \Big)      \bigg]\!\!\!\!\!\!
        \end{align*}%
}}\\
    where $\mathit{un}(A) \coloneqq \bigvee_{a \in A} \big(a \land \bigwedge_{a \neq a' \in A} \neg a'\big)$ asserts that exactly one proposition from $A$ holds. 
    For $\sigma \in \Sigma = 2^\ap$ we write $\sigma$ for the formula $\bigwedge_{a \in \sigma} a \land \bigwedge_{a \not\in \sigma} \neg a$. 
    This formula asserts that the propositions in $V_\tss$ describe a valid path in $\tss$ and the propositions in $V_{\aut^\phi}$ a valid run of $\aut^\phi$ where the first component in $\aut^\phi$ is read as input and the second component is the label of the path described by $V_\tss$.
   
    Similarly, we use propositions $\hat{V}_\tss \coloneqq \{\hat{p}_s \mid s \in S\}$  and $\hat{V}_{\aut^\phi} \coloneqq \{\hat{p}_q \mid q \in Q^\phi\}$ to encode a second path and automaton run (as needed in the definition of $\lproRel{q}{x}{s}$). 
    We define $\widehat{\mathit{valid}}$ analogously to $\mathit{valid}$ but use $\hat{p}_s$ instead of $p_s$ and $\hat{p}_q$ instead of $p_q$.
    Now consider the following QPTL formula $\phi_{q, x, s}$:\\
    \scalebox{0.86}{
        \parbox{1.16\linewidth}{
        \begin{align}
            \!\!\!&\existsu V_\tss \cup V_{\aut^\phi}\ldot \bigg[p_{s}  \land  p_q \land \mathit{valid}  \land  \LTLsquare \LTLdiamond  \bigvee_{q  \in F^\phi} p_q \bigg] \, \land \label{eq:1} \\
            &\quad\Bigg(\forallu \hat{V}_\tss \cup \hat{V}_{\aut^\phi}\ldot \bigg[\Big(\bigvee_{s' \in \sucs{x}} \!\!\!\!\! \hat{p}_{s'} \Big) \land  \hat{p}_q \land \widehat{\mathit{valid}} \land  \LTLsquare \LTLdiamond  \bigvee_{q  \in F^\phi} \hat{p}_q  \bigg] \label{eq:2}\\
            &\quad\quad\rightarrow \Big( \bigvee_{q \not\in F^\phi} \hat{p}_q  \Big) \LTLuntil \Big(\bigvee_{q  \in F^\phi} p_q  \Big) \Bigg) \label{eq:3}
        \end{align}%
}}\\
    This formula closely follows the definition of $\lproRel{q}{x}{s}$.
    We existentially quantify over a path of $\tss$ starting in $s$ and an accompanying run of $\aut^\phi$ (starting in $q$). 
    Taking only (\ref{eq:1}) would result in a direct QPTL formula encoding of the prophecy $\lpro{q}{s}$ from \Cref{sec:compSafety}.
    To encode the optimality, in (\ref{eq:2}) we quantify over an alternative run that starts in some $s' \in \sucs{x}$ and is also accepting. 
    Finally, (\ref{eq:3}) states that the alternative run (described via the $\hat{p}$ propositions) does not visit an accepting state as long as the existentially quantified run has not visited an accepting state. 
    It is easy to see that the QPTL formula $\phi_{q, x, s}$ expresses $\lproRel{q}{x}{s}$.
\end{proof}

\subsection{Correctness Proof}

We show that the resulting set of prophecies is complete.
Let $\proRel{q}{x}{s}$ be a QPTL formula for $\lproRel{q}{x}{s}$ (which exists by \Cref{prop:omegaReagular}).

\begin{restatable}{theorem}{compltnessGeneral}\label{theo:compgeneral}
	Assume $\tss \models \varphi$.
    Define $\Xi = \{ \proRel{q}{x}{s} \mid q \in Q^\phi, s \in \sucs{x}\}$ and let $P = \{ \provRel{q}{x}{s} \mid q \in Q^\phi, s \in \sucs{x}\}$ be a fresh set of atomic propositions. 
    Then $\winss{\veri}{\game{\tss^P}{\varphi^{P, \Xi}}}$.
\end{restatable}
\begin{proofSketch}
    To construct a winning strategy for $\veri$ we use a similar construction as in \Cref{sec:safetyProof}.
    Whenever the $\exists$-player is in a state $x$ and $q$ is the current state of $\aut^\phi$ (reached on the prefix of the game), the strategy checks if any prophecy variable $\provRel{q}{x}{s}$ is set for some $s \in \sucs{x}$ and, if this is the case, selects any such $s$.
    Arguing the correctness of the resulting strategy is more challenging than in the safety case.
    We only sketch the proof here.
    We can assume, that the prophecies are set correctly (so the premise of $\varphi^{P, \Xi}$ holds).
    Under this assumption, we show that there always exists at least one successor state for which the prophecy holds.
    We employ a ranking argument to prove that the resulting play visits $F^\phi$ infinity many times.
    We define a function that maps each $q \in Q^\phi$, $x \in S$, and trace $t$ to an element in $\nat \cup \{\infty\}$ as the shortest number of steps any trace starting in a successor of $x$ needs to take to reach an accepting state. 
    Formally
    \begin{align*}
        \optPath{q}{x}{t} \coloneqq \min\limits_{t' \in \reponse{q}{x}{t}} \firstvisit{F^\phi}{\aut^\phi_q}{t \otimes t'}
    \end{align*}
    where
    \begin{align*}
        \reponse{q}{x}{t} \coloneqq \{ t'  \mid \exists &s \in \sucs{x}\ldot t' \in \traces{\tss_{s}} \; \land \\
        &t \otimes t' \in \calL(\aut^\phi_q) \}.
    \end{align*}
    We can establish that $\mathit{opt}$ serves as a ranking function w.r.t.~our prophecies as follows.
    If $q \not\in F^\phi$, $s \in \sucs{x}$ and $t \in \lproRel{q}{x}{s}$, then $\optPath{q'}{s}{t[1,\infty]} < \optPath{q}{x}{t}$ (where $q' = \delta^\phi(q, (t(0), L(s)))$).
    That is, if a prophecy holds for $s \in \sucs{x}$, then the ranking function is finite and by moving to $s$ the function either decreases \emph{strictly} or an accepting state in $F^\phi$ visited. 
    As $\nat$ is well-founded, this implies that a visit to an accepting state occurs infinity many times. 
    We give a detailed proof in \ifFull{\refApp{sec:appBuchi}}{the full version \cite{fullVersion}}.
\end{proofSketch}

\subsection{Completeness Beyond Deterministic Büchi Automata}\label{sec:beyond}
Up to this point, we assumed that $\aut$ is given as a deterministic Büchi automaton. 
We now sketch how to relax this further.
For this, we assume that $\phi$ is given as a deterministic Rabin automaton (DRA).
In a Rabin automaton, the acceptance condition is given as a set of pairs $(B_1, F_1), \ldots, (B_m, F_m)$. A run $r$ of the automaton is accepting if there exists a $1 \leq i \leq m$ such that $r$ visits states in $B_i$ only finitely many times and states in $F_i$  infinitely many times. As every parity condition is also a Rabin condition, we can translate every LTL formula $\phi$ (or, more generally, any $\omega$-regular property) into an equivalent deterministic Rabin automaton.

\subsubsection{One-pair Rabin Automata}\label{sec:singlePair}

To begin with, we consider the case where $\phi$ can be recognized by a DRA with a \emph{single} pair, i.e., the acceptance condition consists of a set of states $F$ that should be visited infinitely many times and a set of states $B$ that should be visited only finitely many times. 
The previous construction of $\lproRel{q}{x}{s}$ for deterministic Büchi automaton is incomplete as it guarantees that $F$ is visited infinitely many times but does not ensure that $B$ is only visited finitely many times.
We sketch how the definition of $\lproRel{q}{x}{s}$ is modified to support single-pair DRA and refer the reader to \ifFull{\refApp{app:constructionBeyond}}{the full version \cite{fullVersion}} for a formal definition.
We modify $\lproRell{q}{x}{s}$ such that a trace $t$ satisfies $t \in \lproRell{q}{x}{s}$ if either of the following holds:
\begin{itemize}
    \item There exists a trace $t'$ starting in $s$ that is winning (i.e., $t \otimes t' \in \calL(\aut^\phi_q)$), where the unique run \emph{never} visits a state in $B$, and for all other states $s' \in \sucs{x}$ and any winning trace $t''$ starting in $s'$ that also never visits $B$, $t \otimes t'$ visits a state in $F$ at least as fast as $t \otimes t''$, or
    \item There does \emph{not} exist a winning trace that never visits a state in $B$ from any successor of $x$ but there is a trace $t'$ from $s$ that is winning (but visits $B$ at least once), and, for all states $s' \in \sucs{x}$ and any winning trace $t''$ starting in $s'$, the \emph{last} visit to a state in $B$ on $t \otimes t'$ is at least as fast as the last visit on $t \otimes t''$.
\end{itemize}

\noindent
Suppose the $\exists$-player follows the recommendation given by the resulting prophecies (in the sense outlined in the proof sketch of \Cref{theo:compgeneral}).
By doing so, it will construct a witness trace that visits $B$ for the last time as soon as possible and afterward (repeatedly) visits states in $F$ as soon as possible and is therefore winning.

\subsubsection{Beyond One-pair Rabin Automata}

To move from a one-pair DRA to an arbitrary DRA, we simply annotate prophecies with a Rabin pair index. 
Given a DRA with pairs $(B_1, F_1), \ldots, (B_m, F_m)$ we compute the prophecies for a single-pair DRA for each such pair (i.e., the DRA obtained by replacing the set of Rabin pairs with a singleton set). 
For $q \in Q^\phi$, $s \in \sucs{x}$ and $1 \leq i \leq m$ we define $\lproRab{q}{x}{s}{i}$ as the prophecy $\lproRell{q}{x}{s}$ computed on the single-pair Rabin automaton with pair $(B_i, F_i)$ (as in \Cref{sec:singlePair}).
The $\exists$-player can then initially commit to one Rabin pair, say $i$, and afterward, always follow the recommendations of the prophecies where the index matches $i$ (in the sense outlined in the proof sketch of \Cref{theo:compgeneral}).
This strategy constructs a witness trace that is already winning for the DRA fixed to the single pair $(B_i, F_i)$ and therefore also for the general automaton. 
This concludes the proof of \Cref{theo:comp}.

\subsection{On the Number of Prophecies}\label{sec:generllogn}

In our construction, the number of prophecies is linear in the size of the automaton but quadratic in the size of the system. 
More precisely, as we consider prophecies $\lproRel{q}{x}{s}$ where $s \in \sucs{x}$, the number (in the size of the system) is of order $\calO(d \cdot |S|)$ where $d = \max_{s \in S} |\sucs{s}|$ (which is $\mathcal{O}(|S|^2)$ in the worst case).
Using a more efficient binary encoding, we can achieve an exponential decrease in the number of prophecies to $\calO(|S|\log |S|)$ (in the size of the system).
See \ifFull{\refApp{app:nlogn}}{the full version \cite{fullVersion}} for the optimized construction.

\begin{restatable}{proposition}{propnlogn}
	Let $\tss$ be a (finite-state) transition system with state-space $S$ and let $\varphi$ be a $\forall^*\exists^*$ HyperLTL property such that $\tss \models \varphi$.
	There exists a complete set of prophecies $\Xi$ with $|\Xi| \in \calO(|S| \log |S|)$. 
\end{restatable}

\section{Prophecy-based Verification and Implementation}\label{sec:eval}

\subsection{Prophecy-based Verification}\label{sec:mc}

\begin{alg}[t!]
\begin{algorithmic}[1]
    \State \textbf{Input}: $\tss = (S, S_0, \varrho, L)$, $\varphi = \forall \pi. \exists \pi'. \phi$
    \State \text{construct} DPA $\aut^\phi = (Q^\phi, q_0^\phi, \delta^\phi, c^\phi)$
    \For{$i = 0 \dots |Q^\phi|\cdot |S|$}
    \For{$X$ \textbf{in} $2^{Q^\phi \times S}_i$}
    
    \State $\Theta \leftarrow \{ \aut_{\lpro{q}{s}}  \mid (q, s) \in X  \}$
    \State \textbf{if} $\winss{\veri}{\game{\tss^P}{\varphi^{P, \Theta}}}$ \textbf{then} \textbf{return} \cmark{} 
    \EndFor
    \EndFor
    \State \textbf{return} \xmark{}
\end{algorithmic}

\caption{Prophecy-based verification for $\forall\exists$ HyperLTL with safety matrix. The algorithm returns \cmark{} if $\tss \models \varphi$ and \xmark{} if $\tss \not\models \varphi$. We write $2^{S \times Q^\phi}_i$ for all subsets of $S \times Q^\phi$ with cardinality $i$.
	Automaton $\aut_{\lpro{q}{s}}$ represents the prophecy $\lpro{q}{s}$.
	This automaton can be computed by constructing the product of $\aut^\phi$ and $\tss$ and is thus linear in the size of $\tss$.
}\label{alg:mc}
\end{alg}

As the completeness result in this paper is constructive and computable, we directly obtain an algorithmic solution to the HyperLTL model checking problem.
We sketch a possible algorithm for the safety case (cf.~\Cref{sec:compSafety}) in \Cref{alg:mc}.
For each number of prophecies $i$ (ranging from $0$ to $|Q^\phi|\cdot |S|$), we consider all possible sets of prophecies $X$ of size $i$, compute an automaton representation $\aut_{\lpro{q}{s}}$ of $\lpro{q}{s}$ for each $(q, s) \in X$, and check if $\winss{\veri}{\game{\tss^P}{\varphi^{P, \Theta}}}$  holds.
By completeness, the prophecies set of size  $i =|Q^\phi|\cdot |S|$ is complete; the algorithm constitutes a sound-and-complete model checking procedure for $\forall\exists$ properties with safety matrix.\footnote{Of course, computing the set of all prophecies identified in \Cref{theo:compSafety} directly (i.e., immediately setting $i = |Q^\phi|\cdot |S|$) also constitutes a complete model checker. Incrementally increasing the size of the prophecy set (as done in \Cref{alg:mc}) often results in successful verification with fewer prophecies and, in consequence, also in faster computation. If, on the other hand, the goal is to \emph{disprove} a property, computing the full set of prophecies directly is, obviously, more efficient.}  

We briefly discuss how we can check if $\winss{\veri}{\game{\tss^P}{\varphi^{P, \Theta}}}$.
We first observe that we can write $\phi^{P, \Xi}$ (the matrix of $\varphi^{P, \Xi}$) as\\
\scalebox{0.98}{\parbox{\columnwidth}{
	\begin{align*}
		\bigg[\ltlG \Big(\bigwedge_{j=1}^n ({p_j}_{\pi_1} \leftrightarrow {\xi_j})\Big) \to \phi\bigg] \equiv \bigg[\Big(\bigvee_{j=1}^n \ltlF({p_j}_{\pi_1} \not\leftrightarrow {\xi_j})\Big) \lor \phi\bigg].\!\!\!\!\!\!
	\end{align*}
}}\\
In \Cref{alg:mc} we compute an NBA representation $\aut \in \Theta$ for each prophecy. 
We can thus construct an NBA for $\phi^{P, \Theta}$ that is at most exponential in the size of the automata in $\Theta$, convert to a DPA, and solve the parity game $\game{\tss^P}{\varphi^{P, \Theta}}$.\footnote{In particular, we get that \Cref{alg:mc} solves the model checking problem in $2$-\texttt{EXPTIME} in the size of the system.
We emphasize that the goal of our completeness proof is not to derive an efficient model checking algorithm.
As we will see in \Cref{sec:experiments}, the actual number of prophecies needed is often much smaller and research into more efficient prophecy constructions is an interesting direction for future work (cf.~\Cref{sec:futureWork}).
}
Alternatively, we can make use of the disjunctive structure of $\phi^{P, \Xi}$ by constructing a DPA for each formula $\ltlF({p_i}_{\pi_1} \not\leftrightarrow {\xi_i})$ \emph{individually} and then solve a generalized parity game (a game where the winning condition is a disjunction of parity objectives) \cite{ChatterjeeHP07}.

\begin{remark}
	\Cref{alg:mc} uses the prophecy construction for HyperLTL formulas with a safety matrix.
	Analogously, we could obtain a complete algorithm for arbitrary $\forall^*\exists^*$ properties by using the more general prophecy construction in \Cref{sec:compGeneral}.
	However, generating automata representations of the prophecies is more challenging (cf.~\Cref{prop:omegaReagular}).
\end{remark}

\subsection{Implementation and Evaluation}\label{sec:experiments}

We have implemented \Cref{alg:mc} (supporting $\forall^*\exists^*$ properties instead of only $\forall\exists$ properties) in a prototype model checker called \texttt{HyPro} (short for \textbf{Hy}perproperty Verification with \textbf{Pro}phecies).
The novelty of \texttt{HyPro} is twofold:
First, it is the first tool to fully automatically synthesize winning strategies for the $\exists$-player (based on the parity-game-based encoding).
And second, \texttt{HyPro} is the first \emph{complete} verification tool for $\forall^*\exists^*$ properties with an LTL safety matrix.

If desired by the user, \texttt{HyPro} applies a bisimulation-based preprocessing of the system.\footnote{For two bisimilar systems $\tss_1$ and $\tss_2$ (see, e.g., \cite{0020348} for a formal definition) it holds that $\winss{\veri}{\game{\tss_1}{\varphi}}$ iff $\winss{\veri}{\game{\tss_2}{\varphi}}$ for every $\varphi$. 
	Therefore, we can apply strategy-based verification to the (in many cases much smaller) bisimulation quotient.
	Note $\winss{\veri}{\game{\tss_1}{\varphi}}$ and $\winss{\veri}{\game{\tss_1}{\varphi}}$ are, in general, not equivalent when $\tss_1$ and $\tss_2$ are only trace equivalent. }
We have disabled this preprocessing for our experiments. 

\begin{table}[t!]
    
    \caption{Evaluation on instances where no prophecies are necessary to verify a property.
    	We give the problem instance, the bitwidth of the variables (Bitwidth), the size of the program's state space after compiling to a transition system (Size), the verification outcome (Res) (\cmark~indicates that the property holds, \xmark~that it is violated), and the overall time taken by \texttt{HyPro} ($t$). Times are given in seconds.}\label{fig:eval1}
    
    \centering
    \def\arraystretch{1.4}
    \setlength\tabcolsep{3.2mm}
    \begin{tabular}{lcccc}
        \toprule
        \textbf{Instance} & \textbf{Bitwidth} & \textbf{Size} & \textbf{Res} & $\boldsymbol{t}$ \\
        \midrule
        
        \multirow{2}{*}{P1 (GNI)} & $1$-bit &  17   & \multirow{2}{*}{\cmark} & 0.1   \\
         & $4$-bit &  129 & & 25.3   \\
         \arrayrulecolor{black!20}\midrule
        P2 (GNI) & $1$-bit &  55 &  \cmark & 0.4  \\
         \arrayrulecolor{black!20}\midrule
        \multirow{2}{*}{P3 (GNI)} & $1$-bit & 20 & \multirow{2}{*}{\cmark} & 0.2   \\
        & $3$-bit & 80  & & 5.1   \\
        \arrayrulecolor{black!20}\midrule
        \multirow{2}{*}{P4 (GNI)} & $1$-bit & 29 &  \multirow{2}{*}{\cmark} & 0.2   \\
         & $3$-bit & 113   & & 9.2   \\
         \arrayrulecolor{black!20}\midrule
        FlipOutput (Sym) & $7$-bit  & 512  & \cmark & 9.6   \\
        FlipConjunction (Sym)  & $2$-bit  & 80 &  \cmark & 1.3  \\
        Switch (Sym) & $3$-bit  & 144  &\cmark & 4.4   \\
        \arrayrulecolor{black}\bottomrule
    \end{tabular}
    
\end{table}

\begin{table}[t!]
    
    \caption{Evaluation on instances where prophecies are needed to verify a property and instances where a property does not hold. 
    	We give the size of the system (Size), the number of prophecies identified using (an optimized version of) \Cref{theo:compSafety} (\#P), the minimal cardinality of a complete prophecy set (MinP),  the (cumulative) size of the automata used to represent this (minimal) complete prophecy set (SizeP), the construction time of the prophecies ($t_P$), the verification outcome (Res), and the overall time taken by \texttt{HyPro} ($t$). Times are given in seconds.}\label{fig:eval2}
    
    \centering
    \def\arraystretch{1.4}
    \setlength\tabcolsep{1.9mm}
    \begin{tabular}{lccccccc}
        \toprule
        \textbf{Instance} & \textbf{Size} & \textbf{\#P }& \textbf{MinP} & \textbf{SizeP} & $\boldsymbol{t}_{P}$ & \textbf{Res} & $\boldsymbol{t}$ \\
        \midrule
        Predict1Small & 4 & 10 & 1 & 4 & 0.1 & \cmark & 0.3   \\
        Predict1Large & 20 & 42 & 1 & 4 & 0.1 & \cmark & 1.2   \\
        Predict2 & 4 & 20 & 3 & 12& 1.0 & \cmark & 4.7   \\
        Example~\ref{sec:runningExample} & 2 & 6 & 1 & 4 & 0.2 & \cmark & 0.6 \\
        Example~\ref{sec:example} & 7 & 14 & 1 & 6 & 0.1 & \cmark & 0.5 \\
        EnforceManyProph & 4 & 16 & 3 & 12 & 0.8  & \cmark & 14.5   \\
        \Cref{ex:ltlex} & 4 & 20 & 1 & 3 & 0.5 & \cmark & 0.8   \\
        \Cref{ex:multipleUniversalTraces} & 2 & 10 & 1 & 2 & 0.1 & \cmark & 0.3   \\
        PredictLiveness & 4 & 20 & 1 & 2 & 0.3 & \cmark & 0.6   \\
        MissingShift & 4 & 5 & - & - & 0.2 & \xmark & 0.3   \\
        ViolationSimple & 4 & 9 & - & - & 0.4 & \xmark & 0.8   \\
        \bottomrule
    \end{tabular}
\end{table}

\subsubsection{Model Checking without Prophecies}

We begin by evaluating \texttt{HyPro} on instances that do not require any prophecies, i.e., instances where $\game{\tss}{\varphi}$ is won by $\veri$ and so \Cref{alg:mc} already terminates for $i = 0$. 
Our benchmarks consist of information-flow policies in the form of GNI and symmetry constraints (i.e., properties that require that for every trace, there exists one with the opposite outcome) on boolean programs (including those from \cite{BeutnerF21}) with varying bitwidths.

We give the verification results in \Cref{fig:eval1}.
Our results confirm that our direct parity-game-based encoding can successfully synthesize strategies for the $\exists$-player  in systems of medium size.\footnote{Note that the size column in \Cref{fig:eval1} gives the size of an individual system. If we, e.g., verify GNI,  the size of the resulting parity game is cubic in the size of the system (as GNI involves three trace quantifiers).}
If we enable \texttt{HyPro}`s bisimulation-based preprocessing, we can verify properties of significantly larger size, as, with increasing bitwidths, the bisimulation quotient stays small. 
With preprocessing enabled, \texttt{HyPro} can successfully verify systems with up to 55k states within a few seconds. 

We can contrast \texttt{HyPro} with the approach implemented in \texttt{MCHyper} \cite{CoenenFST19,FinkbeinerRS15}.
\texttt{MCHyper} requires an explicit \emph{user-provided} strategy for the $\exists$-player, which reduces hyperproperty verification to the verification of a trace property.
Obviously, strategy synthesis (as done by \texttt{HyPro}) operators on a different scale than strategy verification (as done by \texttt{MCHyper}).
This motivates the coexistence of both tools:
A user can either favor a fully automatic verification using \texttt{HyPro} or provide an explicit strategy and make use of the industrial-strength offered \texttt{MCHyper}.
In the former, the tedious, error-prone, and time-consuming task of writing an explicit strategy by hand is avoided, whereas the latter supports larger systems. 

\subsubsection{Model Checking with Prophecies}

As a second benchmark, we compiled a collection of very small transition systems that cannot be verified without the use of prophecies.
Our benchmarks include programs where non-deterministic choices need to be resolved before the information needed is provided or where predictions on future behavior are demanded.
The results are given in \Cref{fig:eval2}.
None of the existing solvers \cite{FinkbeinerRS15,CoenenFST19,BeutnerF21} can verify any of these instances.
Moreover, based on our completeness result, \texttt{HyPro} is the first tool that can prove that a property does \emph{not} hold.

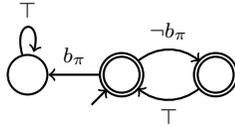
\begin{wrapfigure}{r}{0.38\columnwidth}
	\vspace{-2mm}
	\centering
	\begin{tikzpicture}
		\node[circle, draw,double, thick,minimum size=15pt] at (0,0) (n0) {};
		\node[circle, draw, thick,minimum size=15pt] at (-1.25,0) (n1) {};
		\node[circle, draw,double, thick,minimum size=15pt] at (1.25,0) (n2) {};
		
		\draw[->, thick] (n0) to[out=40,in=140] node[above] {\footnotesize$\neg b_\pi$} (n2);
		\draw[->, thick] (n2) to[out=220,in=-40] node[below] {\footnotesize$\top$} (n0);
		
		\draw[->, thick] (n0) to node[above] {\footnotesize$b_\pi$} (n1);
		
		\path (n1) edge [loop above, thick] node {\footnotesize$\top$} (1);
		
		\draw[->, thick] (n0)+(-0.4, -0.4) -- (n0);
	\end{tikzpicture}

	\caption{Prophecy automaton constructed by \texttt{HyPro} for \Cref{ex:ltlex}.}\label{fig:prop}
\end{wrapfigure}

Even though the prophecies computed by \texttt{HyPro} are only guaranteed to be complete for properties with safety matrix, the construction empirically also works for properties beyond safety (such as  \Cref{ex:ltlex}).
For \Cref{ex:ltlex}, \texttt{HyPro} computes the prophecy depicted as an NBA in \Cref{fig:prop}, which precisely captures the information needed by the $\exists$-player, i.e., it determines if $b_\pi$ never holds at an even position.
Note again that this prophecy is not LTL-definable.  

In \Cref{fig:eval2}, we observe that the actual number of prophecies needed to verify a property (MinP) is often much smaller than the overall number of prophecies (\#P).
This observation is encouraging, as it indicates that the information needed by the $\exists$-player is concise, i.e., expressible with few automata.

We remark that the direct prophecy construction in \Cref{alg:mc} (and implemented in \texttt{HyPro}) is, obviously, limited to very small systems as the number of prophecies scales linearly in the size of the system (also see \Cref{sec:futureWork}). 

\section{Discussion}\label{sec:conclusionFuture}

\subsection{Further Applications of Prophecy-based Verification}

The primary motivation for our work is rooted in the need for efficient and accurate (meaning complete) verification methods for hyperproperties with quantifier alternation. 
Nevertheless, prophecies for hyperproperty verification are also useful beyond just constituting a complete proof method. 
We highlight two such cases in the context of explainable verification results and hyperproperty verification on software.

\subsubsection{Prophecies for Explainable Verification}\label{sec:explain}

Ideally, a verification tool should not only verify that a property holds but convince the user (of, e.g., a security-critical library) \emph{why} this is the case \cite{ChocklerHK08}. 
Certifying verification results of safety trace properties (or $k$-safety hyperproperties) are well understood as the verification tool can provide an (inductive) invariant on the system. 
Understanding verification outcomes in the presence of quantifier alteration, such as for GNI, is much more challenging.
Prophecy-based verification naturally provides a user-understandable certificate.
If a property is verified, a user is provided with (1) an explicit strategy $\sigma$ for the $\exists$-player, (2) an invariant on the plays produced by $\sigma$, and (3) a finite set of prophecies needed by $\sigma$.
This triple allows for a deep investigation into the system as the prophecies directly indicate which future decisions are relevant.
The user can even interactively step through the strategy and prophecies and explore the trace constructed by the strategy. 

\subsubsection{Verification of Infinite-state Systems}

Prophecies are also useful in the context of hyperproperty verification on infinite-state systems. 
For such systems, complementation-based verification is, unsurprisingly, impossible.
In contrast, strategy-based verification is applicable (see, e.g., \cite{BeutnerF22CAV}).
Prophecies can strengthen the $\exists$-player and result in more successful verification instances.

\subsection{Future Work}\label{sec:futureWork}

While \texttt{HyPro} demonstrates that an explicit prophecy construction is applicable in practice, verification is, obviously, restricted to very small systems.
In fact, the direct prophecy-based construction implemented in \texttt{HyPro} is, most likely, easily outperform by complementation-based verification approaches (which are currently not implemented in any tool).
This leaves the construction of more efficient methods to synthesize relevant prophecies as a particularly interesting direction for future work. 
A natural idea would be to, instead of using a \emph{fixed} prophecy construction (as in \Cref{alg:mc}), focus on counter-example guided approaches that iteratively add prophecies by analyzing a spoiling strategy for the $\forall$-player in $\game{\tss^P}{\varphi^{P, \Xi}}$.
Existing techniques for LTL learning \cite{NeiderG18,LemieuxPB15}, or automaton learning \cite{Angluin87,FinkbeinerHT19,FarzanCCTW08} can be used to identify prophecies that distinguish traces on which different future behavior by the $\exists$-player is necessary.
This would exhibit much of the benefits of prophecy-based verification (in particular, the explainability of verification results) while scaling well in the size of the system. 
As we establish with this paper, a well-chosen prophecy generation (that, in the limit, enumerates enough prophecies) would constitute a complete verification method.
Moreover, as demonstrated in \Cref{fig:eval2}, the actual number of prophecies needed in practice is often small.

\section{Conclusion}

In this paper, we have provided a formal footing for the use of prophecy variables for hyperproperty verification by giving a precise characterization of their expressive power.
The main result is that prophecies turn strategy-based verification into a complete verification method for arbitrary $\forall^*\exists^*$ properties.
Our completeness proof is informative in the sense that it provides an explicit, effective, and finite-state-representable ($\omega$-regular) construction of the prophecies.
This new foundation asks for further research to devise prophecy-based (complete) verification methods that scale to larger systems. 

\section*{Acknowledgments}
This work was partially supported by the German Research Foundation (DFG) in project 389792660 (\emph{Foundations of Perspicuous Software Systems}, TRR 248).
R.~Beutner carried out this work as a member of the Saarbrücken Graduate School of Computer Science.

\iffullversion
\else
\balance
\fi

\bibliographystyle{IEEEtran}
{
\interlinepenalty=10000 
\bibliography{references}
}

\iffullversion

\appendices

\renewcommand{\theequation}{\thesection.\arabic{equation}}

\section{Soundness of Strategy-based Verification}\label{sec:appSoundness}

This section is devoted to a proof of \Cref{theo:soundness}:

\soundnessTheo*
\begin{proof}
    Assume that $\sigma$ is a \emph{positional} strategy for $\veri$ that wins from every state in $V_\mathit{init}$.
    We show $\tss \models \varphi$.
    Let $t_1, \ldots, t_k \in \traces{\tss}$ be chosen arbitrary (as in the universal quantification in $\varphi$) and let $p_1, \ldots, p_k \in \paths{\tss}$ be paths that generate those traces (i.e., $L(p_i) = t_i$).
    Define $\vertex_0 \coloneqq ( p_1(0), \ldots, p_k(0))$ which, by definition, is a state in $V_{\mathit{init}}$.
    
    We incrementally construct paths $p_{k+1}, \ldots, p_{k+l}$ (and thereby also traces) for the existentially quantified trace variables as follows.
    We initialize $p_{k+1}, \ldots, p_{k+l} = \epsilon$ (where $\epsilon$ is the empty word).
    For each timestep $x = 0, 1, \ldots$ let 
    \begin{align*}
        \sigma(\vertex_x) = \big\langle(s_1, \ldots, s_k, s_{k+1}, \ldots, s_{k+l}), q, \forall\big\rangle.
    \end{align*}%
    Define for every $k + 1 \leq i \leq k+l$, we extend the finite path $p_i$ by one step by setting $p_i(x) \coloneqq s_i$, i.e., the successor state chosen by $\sigma$.
    We then set
    \begin{align*}
        \vertex_{x+1} \coloneqq \big\langle(p_1(x+1), \ldots, p_k(x+1), p_{k+1}, \ldots,  p_{k+l}), q', \exists\big\rangle
    \end{align*}%
    where $q' = \delta^\phi\big(q, (L(s_1), \ldots, L(s_{k+l})\big)$.
    It is easy to see that for $1 \leq i \leq k$, $s_i \tssto p_i(x+1)$, so $\vertex_{x+1}$ is a successor of $\sigma(\vertex_x)$ in $\game{\tss}{\varphi}$.
    We then repeat with $x \coloneqq x + 1$.
    
    For the constructed runs (which in the limit are infinite), define $t_i \coloneqq L(p_i)$ for $k+1 \leq i \leq k+l$.
    It is easy to see that $[\pi_1 \mapsto t_1, \ldots, \pi_{k+l} \mapsto t_{k+l}] \models \phi$ as in $\game{\tss}{\varphi}$ the automaton $\aut^\phi$ tracks the acceptance of $\phi$ and $\sigma$ is winning. 
    So $t_{k+1}, \ldots, t_{k+l}$ serve as witness traces for $t_1, \ldots, t_k$ and $\tss \models \varphi$ as required.
\end{proof}

\section{Completeness For Safety Matrix}\label{sec:appSafetyCase}

This section is devoted to a detailed proof of \Cref{theo:compSafety}.
That is, we show that the prophecies constructed in \Cref{sec:compSafety} for the case where the matrix is recognizable by a deterministic safety automaton is complete.

Recall that $\aut^\phi = (Q^\phi, q_0^\phi, \delta^\phi, B^\phi)$ is a deterministic safety automaton for $\phi$.
And let $\tss = (S,S_0, \rho, L)$.
For reading convince we recall the definition of our prophecies. For $q \in Q^\phi, s \in S$ we have defined
\begin{align*}
	\lpro{q}{s}\coloneqq \{ t \in \Sigma^\omega \mid \exists t' \in \mathit{Traces}(\tss_s)\ldot  t \otimes t' \in \calL(\aut^\phi_q)  \}.
\end{align*}

\compltnessSafety*

\subsection{Notation}\label{sec:notation}

To make this section as self-contained as possible (and avoid the reader needing to switch between appendix and main body) we repeat some importation definitions from the main part (which agrees with the notation in \Cref{sec:safetyNotation}). 
By definition $\tss^P$, nodes in $\game{\tss^P}{\varphi^{P, \Xi}}$ either have the form $(s, A)$, where $s \in S$ and $A \subseteq P$ or of the form $\gamenode{(s, A), (s', A'), q, \flat}$ where $s, s' \in S$, $A, A' \subseteq P$ and $\flat \in \{\forall, \exists\}$. 
Here $q$ is an automaton state in a DPA tracking
\begin{align}\label{eq:appLTL}
	\phi^{P, \Xi} \coloneqq \ltlG \Big(\!\!\! \bigwedge_{q \in Q^\phi, s \in S} \!\!\! (\prov{q}{s})_{\pi}  \leftrightarrow \pro{q}{s}\Big)  \rightarrow \phi.
\end{align}%
It is easy to see that in states of the form $\gamenode{(s, A), (s', A'), q, \flat}$ the $A'$ component (stemming from the definition of $\tss^P$) is irrelevant as in (\ref{eq:appLTL}) the proposition in $P$ are only referred to on trace variable $\pi$. 
We therefore replace a node $\gamenode{(s, A), (s, A'), q, \flat}$ with $\gamenode{(s, A), s', q, \flat}$.
With this conceptual simplification, any finite or infinite play $\vertex = \vertex(0)\vertex(1) \cdots \in V^* \cup V^\omega$ in $\game{\tss^P}{\varphi^{P, \Xi}}$ starting in $V_\mathit{init}$ has the form 
\begin{align}\label{eq:appexamplePath}
    \begin{split}
        &(s_0, A_0) \to \gamenode{(s_0, A_0), s_0', q_0, \forall} \to \gamenode{(s_1, A_1), s_0', q_1, \exists}\\
        &\to \gamenode{(s_1, A_1), s_1', q_1, \forall} \to \gamenode{(s_2, A_2), s_1', q_2, \forall} \to \cdots
    \end{split}
\end{align}
We can extract from $\vertex$ both paths through $\tss$ and the prophecies set at each step.
Define $\sSeq{0}\ssSeq{1} \cdots$ to be the path of the $\forall$-copy ($s_0s_1\cdots$ in (\ref{eq:appexamplePath})), $\bSeq{0}\bSeq{1}\cdots$ the sequence of prophecies chosen ($A_0A_1 \cdots$ in (\ref{eq:appexamplePath})) and $\ssSeq{0}\ssSeq{1}\cdots$ the path for the $\exists$-copy ($s_0's_1'\cdots$ in (\ref{eq:appexamplePath})).
Note that for any finite $\vertex$ of odd length (i.e., $|\vertex| = 2i + 1$), $\sSeq{k}, \bSeq{k}$ are defined for all $0 \leq k \leq i$ and $\ssSeq{k}$ for all $0 \leq k \leq i - 1$.
Define $\sTr{k} \coloneqq L(\sSeq{k})$ and $\ssTr{k} \coloneqq L(\ssSeq{k})$ for $0 \leq k \leq i-1$.

\begin{strat}[t!]
	\begin{algorithmic}[1]
		\State \textbf{Input:} $\vertex \in V^*$ \textbf{with} $|\vertex| = 2i + 1$
		\If{$i = 0$}
		\State  $T \coloneqq S_0$     \label{line:app1}
		\Else
		\State  $T \coloneqq \sucs{\ssSeq{i-1}}$\label{line:app2}\vspace{-0.8mm}
		\EndIf
		\State $\hat{q} \coloneqq {\delta^\phi}^* \big[(\sTr{0}, \ssTr{0}) \cdots (\sTr{i-1},\ssTr{i-1})\big]$\label{line:app3}
		\State $C \coloneqq \{ s' \mid s' \in T \land  \prov{\hat{q}}{s'} \in \bSeq{i} \}$\label{line:app4}
		\If{$C \neq \emptyset$}
		\State \textbf{return} any $s' \in C$\label{line:app5}
		\Else
		\State \textbf{return} any $s' \in T$\label{line:app6}
		\EndIf
	\end{algorithmic}
	
	\caption{Winning strategy for $\veri$. Identical to the construction in \Cref{sec:safetyProof} and repeated here for the reader's convenience.}\label{fig:appstrategy}
\end{strat}

\subsection{Strategy Construction}

With those definitions at hand, we can describe an explicit winning strategy for $\veri$ in \Cref{fig:appstrategy}.
Note that $\sigma$ returns a successor for the $\exists$-copy which already gives a unique successor vertex in the game (as we identified states $\gamenode{(s, A), (s, A'), q, \flat}$ with $\gamenode{(s, A), s', q, \flat}$). 

By the structure of the parity game, any finite path starting in $V_{\mathit{init}}$ that reaches a node in $V_\veri$ is of odd length.
We begin by computing all possible successor states for the $\exists$-copy in a set $T$. These are either all initial nodes in case where $|\vertex| = 1$ (line \ref*{line:app1}) or all successor states of the current state of the $\exists$-copy (line \ref*{line:app2}).
We then compute the current state of $\aut^\phi$ reached on $\vertex$ in line \ref*{line:app3}.
Note that $\hat{q}$ is a state in $\aut^\phi$ whereas the automaton states occurring in $\vertex$ are states in a DPA for (\ref{eq:appLTL}).
In line \ref*{line:app4}, we check if any of the possible successors in $T$ are declared safe in the sense that the prophecy is set for that state.
If there is any such state we pick it (line \ref*{line:app5}).
Otherwise, we choose an arbitrary successor (line \ref*{line:app6}).

\subsection{Invariant}

We now claim that the constructed strategy is winning. 
For this, let $\vertex  \in V^\omega$ be any infinite play in $\game{\tss^P}{\varphi^{P, \Xi}}$ starting in $V_{\mathit{init}}$ that is compatible with $\sigma$. 
We fix $\vertex$ throughout this section.
We want to show that $\vertex$ is winning for $\veri$.
By the construction of $\game{\tss^P}{\varphi^{P, \Xi}}$, this is equivalent to the fact that for the two traces $\trs \coloneqq \sTr{0}\sTr{1}\cdots$ and $\trss \coloneqq \ssTr{0}\ssTr{1}\cdots$  it holds that $[\pi \mapsto \trs, \pi' \mapsto \trss] \models \phi^{P, \Xi}$ where $\phi^{P, \Xi}$ is the formula given in (\ref{eq:appLTL}).
We can therefore assume that $\trs$ (the trace constructed by the $\forall$-player) satisfies the LTL property (the premise of $\phi^{P, \Xi}$)
\begin{align}\label{eq:headLTL}
    \ltlG \big(\bigwedge_{q \in Q^\phi, s \in S} \prov{q}{s} \leftrightarrow \pro{q}{s}\big) 
\end{align}%
as otherwise $\vertex$ trivially satisfies (\ref{eq:appLTL}) and is thus won by $\veri$.
Phrased differently, this amounts to the following assumption:

\begin{assumption}\label{as:ma}
    For any $q \in Q^\phi$, $s \in S$, and any $i$ we have that $\trs[i, \infty] \in \lpro{q}{s}$ if and only if $\prov{q}{s} \in \bSeq{i}$. 
\end{assumption}

Under this assumption we still need to show that $[\pi \mapsto \trs, \pi' \mapsto \trss] \models \phi$ (the conclusion of $\phi^{P, \Xi}$).

For each $i \in \nat$ define 
\begin{align*}
    \qSeqHat{i} \coloneqq  {\delta^\phi}^* \big[(\sTr{0}, \ssTr{0}) \cdots (\sTr{i-1},\ssTr{i-1})\big]
\end{align*}
Note that $\qSeqHat{0} = q^\phi_0$. 
We can observe two things:
First, the sequence $\qSeqHat{0}\qSeqHat{1}\qSeqHat{2} \cdots$ is the unique run of  $\aut^\phi$ on $\trs \otimes \trss$.
And second, $\qSeqHat{i}$ is the same state computed by $\sigma$ in line \ref*{line:app3} of \Cref{fig:appstrategy} in the $i$th iteration (as $\vertex$ is compatible with $\sigma$). 

We establish the following invariant:

\begin{proposition}\label{prop:claim}
    For every $i \in \nat$ it holds that 
    \begin{align*}
    	\prov{\qSeqHat{i}}{\ssSeq{i}} \in \bSeq{i}.
    \end{align*}
\end{proposition}

Expressed less formally, in every step $i$ the state selected by $\sigma$ (which is $\ssSeq{i}$) is such that the prophecy for that state (and automaton state $\qSeqHat{i}$) holds.
In particular, the set $C$ computed in line \ref*{line:app4} of \Cref{fig:appstrategy} is never empty.\footnote{The set $C$ might be empty if we consider an arbitrary run of $\sigma$ which is why we include the special treatment in line \ref*{line:app6}.  
In our proof, we assume \Cref{as:ma}.
For such plays $C \neq \emptyset$ in every iteration and line \ref*{line:app6} is never reached.
}
Before proving \Cref{prop:claim}, we show that the statement in \Cref{prop:claim} already suffices to show the desired result.

\begin{lemma}
    Play $\vertex$ is won by $\veri$.
\end{lemma}
\begin{proof}
    By the structure of the game, it suffices to show that $[\pi \mapsto \trs, \pi' \mapsto \trss] \models \phi$ (the conclusion of $\phi^{P, \Xi}$).
    By definition $\qSeqHat{0}\qSeqHat{1}\cdots$ is the unique run of $\aut^\phi$ on $\trs \otimes \trss$. 
    As $\aut^\phi$ is a safety automaton it is sufficient to show that for every $i$, $\qSeqHat{i}$ is not a bad state. 
    It is easy to see that if $q$ is a bad state in $\aut^\phi$ then $\lpro{q}{s}$ is empty for all $s \in S$.
    According to \Cref{prop:claim}, $\prov{\qSeqHat{i}}{\ssSeq{i}} \in \bSeq{i}$ so by \Cref{as:ma} we get that $\trs[i, \infty] \in \lpro{\qSeqHat{i}}{\ssSeq{i}}$.
    Consequently, $\lpro{\qSeqHat{i}}{\ssSeq{i}}$ is non-empty and so $\qSeqHat{i}$ cannot be bad.
\end{proof}

As this holds for any play $\vertex$ compatible with $\sigma$, we have a proof of \Cref{theo:compSafety}.

\subsection{Proof of the Invariant}

It remains to show the invariant \Cref{prop:claim} .
For our proof, we need the following two lemmas.
They establish that there initially exists at least one prophecy that is set to true and if the prophecy recommendation is followed, there exists at least one prophecy set in the next step.

\begin{lemma}\label{lem:indStart}
    There exists a state $s'_0 \in S_0$ such that $\trs \in \lpro{q^\phi_0}{s'_0}$.
\end{lemma}
\begin{proof}
    By assumption we have that $\tss \models \forall \pi.\exists \pi'. \phi$.
    As $\trs \in \traces{\tss}$ we can plug it into the universal quantifier and get $t' \in \traces{\tss}$ such that $\trs \otimes t' \in \calL(\aut^\phi)$.
    Let $p' \in \paths{\tss}$ be such that $L(p') = t'$.
    Now define $s_0' \coloneqq p'(0)$ which satisfies $s_0' \in S_0$.
    We claim that $\trs \in \lpro{q^\phi_0}{s'_0}$.
    To show this, we can simply plug $t'$ in for the existential quantifier in the definition of $\lpro{q^\phi_0}{s'_0}$.
    By assumption this witness satisfies $\trs \otimes t' \in \calL(\aut^\phi) = \calL(\aut^\phi_{q_0^\phi})$.
\end{proof}

\begin{lemma}\label{lem:progress}
    If $t \in \lpro{q}{s}$ then there exists a state $x \in \sucs{s}$ such that $t[1, \infty] \in \lpro{q'}{x}$ where $q' = \delta^\phi(q, (t(0), L(s)))$.
\end{lemma}
\begin{proof}
    Assume $t \in \lpro{q}{s}$. 
    By definition of $\lpro{q}{s}$ there exists a $t' \in \traces{\tss_s}$ such that $t \otimes t' \in \calL(\aut^\phi_q)$.
    Let $p' \in \paths{\tss_s}$ be such that $L(p') = t'$. 
    Now define $x \coloneqq p'(1)$. 
    We claim that $x$ serves as a witness to the current lemma.
    Obviously, $x \in \sucs{s}$ as $p'(0) = s$ and $p'$ is a valid path in $\tss$.
    We need to show that $t[1, \infty] \in \lpro{q'}{x}$.
    As the witness for the existential quantification in the $\lpro{q'}{x}$ we select $t'' \coloneqq t'[1,\infty]$.
    As $p'' \coloneqq p'[1,\infty]$ satisfies $L(p'') = t''$ and $p''(0) = (p'[1,\infty])(0)  = p'(1) = x$ we have $t'' \in \traces{\tss_x}$ as required.
    It remains to show that $t[1, \infty] \otimes t'' \in \calL(\aut^\phi_{q'})$.
    By assumption we have $t \otimes t' \in \calL(\aut^\phi_q)$. 
    Now
    \begin{align*}
        t\otimes t'  \in \calL(\aut^\phi_{q})
        &\Rightarrow(t(0) t[1, \infty]) \otimes (L(s)t'') \in \calL(\aut^\phi_{q}) \\
        &\Rightarrow (t(0), L(s)) (t[1, \infty] \otimes t'') \in \calL(\aut^\phi_{q}) \\
        &\Rightarrow t[1, \infty] \otimes t'' \in \calL(\aut^\phi_{ \delta^\phi(q, (t(0), L(s)))}) \\
        &\Rightarrow t[1, \infty] \otimes t'' \in \calL(\aut^\phi_{q'})
    \end{align*}%
    where the first implication holds by definition of $t''$, the second  by definition of $\otimes$, the third by the automaton semantics and the forth by the definition unrolling of $q'$. 
\end{proof}

We are now in a position to prove \Cref{prop:claim}.

\begin{proof}[Proof of \Cref{prop:claim}]
    The proof is by induction on $i$.\\
    \textbf{The case $i = 0$:}
      By \Cref{lem:indStart} there exists $s'_0 \in S_0$ such that $\trs \in \lpro{q^\phi_0}{s'_0}$.
        Now by \Cref{as:ma} this implies that $\prov{q^\phi_0}{s_0'} \in \bSeq{0}$.
        For the set $C$ computed in iteration $i = 0$ in line \ref*{line:app4} of \Cref{fig:appstrategy} we have 
        \begin{align*}
            C = \{ \prov{\qSeqHat{0}}{s} \mid \prov{\qSeqHat{0}}{s} \in  \bSeq{0} \land s \in S_0  \}.
        \end{align*}
        In particular, $\prov{q^\phi_0}{s_0'} \in C$ so $C \neq \emptyset$. 
        By construction $\sigma$ will have picked some initial state $\ssSeq{0} \in C$ (not necessarily $s'_0$) which already implies that $\prov{\qSeqHat{0}}{\ssSeq{0}}  \in \bSeq{0}$ as required.
        
       \textbf{The case $i > 0$:}
         By the induction hypothesis we can assume that $\prov{\qSeqHat{i-1}}{\ssSeq{i-1}}  \in \bSeq{i-1}$. 
        By \Cref{as:ma}, this implies that 
        \begin{align*}
            \trs[i-1, \infty] \in \lpro{\qSeqHat{i-1}}{\ssSeq{i-1}}.
        \end{align*}
        By \Cref{lem:progress} there exists $x \in \sucs{\ssSeq{i-1}}$ such that
        \begin{align*}
            (\trs[i-1, \infty])[1, \infty] = \trs[i, \infty] \in \lpro{q'}{x}
        \end{align*}
        where 
        \begin{align*}
            q' = {\delta^\phi}^*(\qSeqHat{i-1}, (\trs(i-1), L(\ssSeq{i-1}))) = \qSeqHat{i}.
        \end{align*}
        As $\trs[i, \infty] \in \lpro{\qSeqHat{i}}{x}$ we can use \Cref{as:ma} again and obtain that $\prov{ \qSeqHat{i}}{x} \in \bSeq{i}$.
        Now by construction of $\sigma$ this implies that $x \in C$ (where $C$ is the set computed by $\sigma$ in the $i$th iteration in line \ref*{line:app4}) so $C \neq \emptyset$ and the strategy will have picked some successor $\ssSeq{i} \in C$ (not necessarily $x$) which already implies that $\prov{\qSeqHat{i}}{\ssSeq{i}} \in \bSeq{i}$ as required.
\end{proof}

This concludes the proof of \Cref{prop:claim} and therefore the proof of \Cref{theo:compSafety}.

\section{Lower Bounds on the Number of Prophecies}\label{app:lowerBound}

This section is devoted to a proof of \Cref{prop:lb}.

We note that if both the size of the system and specification can depend on $n$, deriving a $\log n$ lower bound on the number of prophecies needed is easy, as stated in \Cref{lem:simpleLowerBound}.

\begin{lemma}\label{lem:simpleLowerBound}
	There exists a family of transition systems $\{\tss_n\}_{n \in \nat}$ and $\forall\exists$ HyperLTL properties $\{\varphi_n\}_{n \in \nat}$ such that $\tss_n$ has $\Theta(n)$-many states, and $\varphi_n$ has a safety matrix, and $\varphi_n$ has size $\Theta(n)$, and $\tss_n \models \varphi_n$  and, additionally, any family of prophecies $\{\Xi_n\}_{n \in \nat}$ where where $\Xi_n$ is complete for $\tss_n, \varphi_n$ (i.e., $\winss{\veri}{\game{\tss_n^P}{\varphi_n^{P, \Xi_n}}}$ for some fresh set $P$) has at last size $|\Xi_n| \in \Omega(\log n)$.
\end{lemma}
\begin{proof}
	For every $n \in \nat$, define the transition system $\tss_n = (S_n, S_{n, 0}, \rho_n, L_n)$ over $\ap_n \coloneqq \{a_1, \ldots, a_n\}$ by $S_n \coloneqq S_{n, 0} \coloneqq \{s_1, \ldots, s_n\}$, $\rho_n \coloneqq \{(s_i, s_j) \mid 1 \leq i, j \leq n\}$, and $L_n(s_i) \coloneqq \{a_i\}$. 
	That is $\tss_n$ is simply the fully connected system generating all traces over $\ap_n$.
	Define $\varphi_n \coloneqq \forall \pi. \exists \pi'. (\bigwedge_{i=1}^n (a_i)_{\pi'} \leftrightarrow \ltlN (a_i)_{\pi})$.
	Formula $\varphi_n$ requires the $\exists$-player to choose the state $s_i$ that the $\forall$-player chooses in the next move. 
	When choosing the initial state, the $\exists$-player has $n$ possible options of which exactly one is the correct (winning) choice.
	The prophecies thus need to communicate information of size $n$ for which at least $\log n$-bits are required. 
	As each prophecy can only provide a single bit of information the lower bound follows. 
\end{proof}

\begin{figure}[!t]
	\centering
	\begin{tikzpicture}[scale=0.95]
		\node[draw, rectangle, thick,align=center, rounded corners=5pt] at (3.5,2) (bdelay) {$s_\mathit{delay}$\textbf{:}\\$a = \top,d=\top$};
		
		\node[draw, rectangle, thick,align=center, rounded corners=5pt] at (3.5,0.5) (binit) {$s_\mathit{init}$\textbf{:}\\$a = \top$};
		
		\node[draw, rectangle, thick,align=center, rounded corners=5pt] at (0,-1.5) (b11) {$s^1_1$\textbf{:}\\$a = \beta_1^1$};
		\node[draw, rectangle, thick,align=center, rounded corners=5pt] at (2,-1.5) (b12) {$s^2_1$\textbf{:}\\$a = \beta_1^2$};
		\node[draw, rectangle, thick,align=center, rounded corners=5pt] at (4,-1.5) (b13) {$s^3_1$\textbf{:}\\$a = \beta_1^3$};
		\node[] at (5.5,-1.5) () {\Large $\cdots$};
		\node[draw, rectangle, thick,align=center, rounded corners=5pt] at (7,-1.5) (b14) {$s^n_1$\textbf{:}\\$a = \beta_1^n$};

		\node[draw, rectangle, thick,align=center, rounded corners=5pt] at (0,-3) (b21) {$s^1_2$\textbf{:}\\$a = \beta_2^1$};
		\node[draw, rectangle, thick,align=center, rounded corners=5pt] at (2,-3) (b22) {$s^2_2$\textbf{:}\\$a = \beta_2^2$};
		\node[draw, rectangle, thick,align=center, rounded corners=5pt] at (4,-3) (b23) {$s^3_2$\textbf{:}\\$a = \beta_2^3$};
		\node[] at (5.5,-3) () {\Large $\cdots$};
		\node[draw, rectangle, thick,align=center, rounded corners=5pt] at (7,-3) (b24) {$s^n_2$\textbf{:}\\$a = \beta_2^n$};

		\node[draw, rectangle, thick,align=center, rounded corners=5pt] at (0,-4.5) (b31) {$s^1_3$\textbf{:}\\$a = \beta_3^1$};
		\node[draw, rectangle, thick,align=center, rounded corners=5pt] at (2,-4.5) (b32) {$s^2_3$\textbf{:}\\$a = \beta_3^2$};
		\node[draw, rectangle, thick,align=center, rounded corners=5pt] at (4,-4.5) (b33) {$s^3_3$\textbf{:}\\$a = \beta_1^3$};
		\node[] at (5.5,-4.5) () {\Large $\cdots$};
		\node[draw, rectangle, thick,align=center, rounded corners=5pt] at (7,-4.5) (b34) {$s^n_3$\textbf{:}\\$a = \beta_3^n$};
		
		\node[] at (0,-5.7) (b41) {};
		\node[] at (2,-5.7) (b42) {};
		\node[] at (4,-5.7) (b43) {};
		\node[] at (7,-5.7) (b44) {};
		
		\node[] at (0,-6.5) () {\Large $\vdots$};
		\node[] at (2,-6.5) () {\Large $\vdots$};
		\node[] at (4,-6.5) () {\Large $\vdots$};
		\node[] at (7,-6.5) () {\Large $\vdots$};
		
		\node[] at (0,-7.3) (bmm1) {};
		\node[] at (2,-7.3) (bmm2) {};
		\node[] at (4,-7.3) (bmm3) {};
		\node[] at (7,-7.3) (bmm4) {};

		\node[draw, rectangle, thick,align=center, rounded corners=5pt] at (0,-8.5) (bm1) {$s^1_m$\textbf{:}\\$a = \beta_m^1$};
		\node[draw, rectangle, thick,align=center, rounded corners=5pt] at (2,-8.5) (bm2) {$s^2_m$\textbf{:}\\$a = \beta_m^2$};
		\node[draw, rectangle, thick,align=center, rounded corners=5pt] at (4,-8.5) (bm3) {$s^3_m$\textbf{:}\\$a = \beta_m^3$};
		\node[] at (5.5,-8.5) () {\Large $\cdots$};
		\node[draw, rectangle, thick,align=center, rounded corners=5pt] at (7,-8.5) (bm4) {$s^n_m$\textbf{:}\\$a = \beta_m^n$};

		\node[draw, rectangle, thick,align=center, rounded corners=5pt] at (3.5,-10.5) (bend) {$s_\mathit{end}$\textbf{:}\\$a = \top$};
		
		\draw[->, thick] (binit)+(-1,0) -- (binit);
		
		\draw[->, thick] (bdelay)+(-1.6,0) -- (bdelay);
		
		\path (bdelay) edge [thick, ->] (binit);
		
		\path (binit) edge [thick, ->] (b11);
		\path (binit) edge [thick, ->] (b12);
		\path (binit) edge [thick, ->] (b13);
		\path (binit) edge [thick, ->] (b14);
		
		\path (b11) edge [thick, ->] (b21);
		\path (b12) edge [thick, ->] (b22);
		\path (b13) edge [thick, ->] (b23);
		\path (b14) edge [thick, ->] (b24);
		
		\path (b21) edge [thick, ->] (b31);
		\path (b22) edge [thick, ->] (b32);
		\path (b23) edge [thick, ->] (b33);
		\path (b24) edge [thick, ->] (b34);
		
		\path (b31) edge [thick, ->,dashed] (b41);
		\path (b32) edge [thick, ->,dashed] (b42);
		\path (b33) edge [thick, ->,dashed] (b43);
		\path (b34) edge [thick, ->,dashed] (b44);

		\path (bmm1) edge [thick, ->,dashed] (bm1);
		\path (bmm2) edge [thick, ->,dashed] (bm2);
		\path (bmm3) edge [thick, ->,dashed] (bm3);
		\path (bmm4) edge [thick, ->,dashed] (bm4);
		
		\path (bm1) edge [thick, ->] (bend);
		\path (bm2) edge [thick, ->] (bend);
		\path (bm3) edge [thick, ->] (bend);
		\path (bm4) edge [thick, ->] (bend);

		\path (bend) edge [loop below, thick, ->] (bend);
	\end{tikzpicture}
	
	\caption{Transition system $\tss_n$ used in the proof of \Cref{prop:lb}. Any complete set of prophecies for this system (combined with the property in the proof of \Cref{prop:lb}) has size at least $\log n$.
		For proposition $a$ we explicit give the truth value in every state.
		Proposition $d$ is always set to $\bot$ except in state $s_\mathit{delay}$ (where it is explicitly set to $\top$).}\label{fig:lowerBoundTs}
\end{figure}
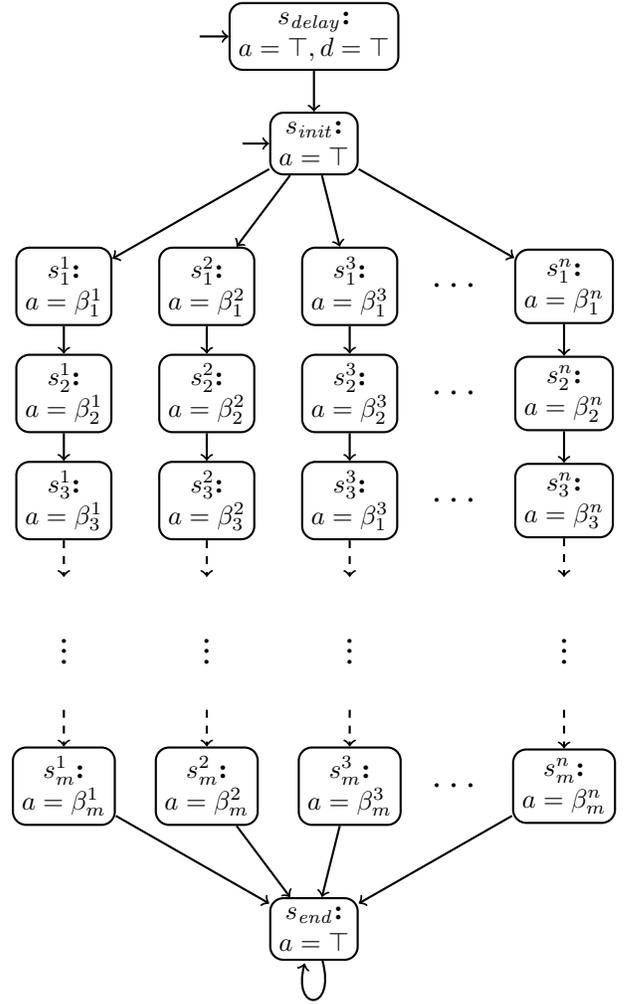

We now use similar ideas as in the above proof, but fix the size of the specification to prove \Cref{prop:lb}.

\lognprop*
\begin{proof}
	To ease the presentation we construct a slightly larger family of transition systems, i.e., a family $\{\tss_n\}_{n \in \nat}$, where $\tss_n$ has $\Theta(n \log n)$-many states.
	We discuss why this suffices at the end of this proof.
	
	Fix any $n \in \nat$. 
	We assume, w.l.o.g., that $n = 2^m$ (and so $m = \log n$).
	The proof idea is similar to that used in \Cref{lem:simpleLowerBound}, i.e., we give the $\exists$-player $n$ options to choose from of which exactly one is the correct one.
	In contrast to \Cref{lem:simpleLowerBound} we fix the size of the property (and thereby also of the number of atomic propositions).
	
	Let $\ap = \{a, d\}$.
	The idea is to use proposition $a$ to encode a binary sequence of length $m$.
	The $\exists$-player can then pick any of those sequences (of which there are $2^m = n$ many) and we require that the $\exists$-player picks the same sequence that the $\forall$-player is choosing. 
	Let $\beta^1, \ldots, \beta^n \in \mathbb{B}^m$ be the $n$ distinct binary sequences of length $m$.
	We write $\beta^i_j$ for the $j$th bit (starting at index $1$) in the $i$th sequences.
	Now consider the transition system $\tss_n$ depicted in \Cref{fig:lowerBoundTs}.
	By moving to a state $s^1_1, \ldots, s^n_1$ any of the $n$ sequences $\beta^1, \ldots, \beta^n$ can be generated. 
	The idea is that the $\forall$-player starts in state $B_\mathit{delay}$ and the $\exists$-player in state $B_\mathit{init}$, so the $\exists$-player has to commit to some state $s^1_1, \ldots, s^n_1$ \emph{before} the $\forall$-player does. 
	We use proposition $d$ to identify the delay state. 
	
	Define the HyperLTL property $\varphi$ as 
	\begin{align*}
		\forall \pi. \exists \pi' \ldot d_\pi \rightarrow \Big( \neg d_{\pi'} \land \ltlN \ltlG (a_{\pi'} \leftrightarrow \ltlN a_\pi  ) \Big).
	\end{align*}
	This formula expresses that for all traces starting in $s_\mathit{delay}$ there exists a trace starting in $s_\mathit{init}$ that traversed the same binary sequence (shifted by one position).
	Note that $\varphi$ is independent of $n$.
	
	It is easy to see that $\tss_n \models \varphi$ as the trace $\pi'$ can simply traverse the same sequence as $\pi$.
	However, to verify this using strategy-based verification, the $\exists$-player needs to fix a state $s^1_1, \ldots, s^n_1$ of which exactly one is the correct (winning) choice.
	Similar to \Cref{lem:simpleLowerBound}, we require at least $\log n$ bits to communicate this.
	Any family of prophecies $\{\Xi_n\}_{n \in \nat}$ where $\Xi_n$ is complete for $\tss_n, \varphi_n$ for every $n$, must thus have size $|\Xi_n| \in \Omega(\log n)$.
	
	For the family $\{\tss_n\}_{n \in \nat}$ constructed above each $\tss_n$ has $\Theta(n \cdot m) = \Theta(n\log n)$-many states. 
	To get a family $\{\tss'_n\}_{n \in \nat}$ where $\tss'_n$ has $\Theta(n)$-many states (as required in the statement we are proving) we use the following:
	Define the family $\{\tss'_n\}_{n \in \nat}$ where $\tss'_n \coloneqq \tss_{\lceil h^{-1}(n) \rceil}$, where $h^{-1}$ is the inverse of $h(n) \coloneqq n \log n$. 
	As we have shown above (for family $\{\tss_n\}_{n \in \nat}$), each  family of prophecies $\{\Xi_n\}_{n \in \nat}$ where $\Xi_n$ is complete for $\tss'_n, \varphi_n$ for every $n$, must thus have size $|\Xi_n| \in \Omega(\log (\lceil h^{-1}(n) \rceil))$.
	To get the desired bound we simply need to check that $\Omega(\log (\lceil h^{-1}(n) \rceil))$ is still in $\Omega(\log n)$.
	For this, note that $\lceil h^{-1}(n) \rceil \geq \lceil \sqrt{n} \rceil$ and $\log (\sqrt{n}) = \tfrac{1}{2} \log n$.
\end{proof}

\begin{remark}
	Recall that we do not claim that our construction using linearly many prophecies (cf.~\Cref{theo:compSafety}) is optimal.
	For the family constructed in \Cref{prop:lb} (and also in \Cref{lem:simpleLowerBound}) $\log n$ is also an upper bound on the number of prophecies needed as we can encode in binary which of the available moves is safe to play. 
	This binary encoding is, in general, not applicable in our construction as the prophecies are oblivious to the current state of the $\exists$-player. 
    Any fixed state identified in a prophecy must thus not be a successor of the $\exists$-player's current state. 
	We discuss this further for the general case \Cref{sec:generllogn}, where we can use binary encodings to gain an exponential improvement in the number of prophecies needed. 
\end{remark}

\section{Completeness for Deterministic Büchi Matrix}\label{sec:appBuchi}

This section is devoted to a proof of \Cref{theo:compgeneral}.
That is, we show that the prophecies constructed in \Cref{sec:compGeneral} for the case where the matrix is recognizable by a deterministic Büchi automaton is complete.

Recall that $\aut^\phi = (Q^\phi, q_0^\phi, \delta^\phi, F^\phi)$ is a deterministic Büchi automaton for $\phi$.
And let $\tss = (S, \{s_0\}, \rho, L)$ (we assumed that $\tss$ has a unique initial state).
For this proof, we also assume that $q_0^\phi \not\in F^\phi$ (this simplifies the proof and can be ensured easily).
For reading convince we recall the definition of our prophecies $\lproRel{q}{x}{s}$ which are defined as:
 \begin{align*}
	\bigg\{ t &\in \Sigma^\omega \mid \exists t' \in \mathit{Traces}(\tss_{s})\ldot t \otimes t' \in \calL(\aut^\phi_q) \,\land \\
	&\Big[\forall s' \in \sucs{x}\ldot \forall t'' \in \mathit{Traces}(\tss_{s'})\ldot t \otimes t'' \in \calL(\aut^\phi_q)\\
	&\Rightarrow \firstvisit{F^\phi}{\aut^\phi_q}{t \otimes t'}  \leq  \firstvisit{F^\phi}{\aut^\phi_q}{t \otimes t''}\Big] \bigg\}
\end{align*}

\compltnessGeneral*

\subsection{Strategy Construction}

We use the same notation that we introduced in \Cref{sec:notation}.
Similar to the safety case in \Cref{sec:compSafety} we give an explicit winning strategy for $\veri$.
We give the strategy in \Cref{fig:strategyBuchi}.
We first compute the set of possible successors in \ref*{line:b3} and then check if there is some successor from the current state (which is $\ssSeq{i-1}$) is set by the prophecies (in \ref*{line:b4}).
The main difference to the safety construction is that the prophecies are now annotated with an additional state from where to start, which is why we use the fixed state $\ssSeq{i-1}$ in \ref*{line:b4}.
As before the main idea of this strategy is to simply follow the recommendation given by prophecies. 

We now claim that the constructed strategy is winning. 
For this, let $\vertex  \in V^\omega$ be any infinite path starting in $V_{\mathit{init}}$ that is compatible $\sigma$. 
We fix $\vertex$ throughout this section.
We want to show that $\vertex$ is winning for $\veri$, i.e., for $\trs \coloneqq \sTr{0}\sTr{1}\cdots$ and $\trss \coloneqq \ssTr{0}\ssTr{1}\cdots$  it holds that $[\pi \mapsto \trs, \pi' \mapsto \trss] \models \phi^{P, \Xi}$, where
\begin{align}\label{eq:appfullLTL}
     \phi^{P, \Xi} \coloneqq \ltlG \Big(\bigwedge_{q \in Q^\phi, x, s \in S} (\provRel{q}{x}{s})_{\pi}  \leftrightarrow \proRel{q}{x}{s}\Big)   \rightarrow \phi.
\end{align}%
As in the safety case, we can therefore assume that the premise of $\phi^{P, \Xi}$ holds:

\begin{assumption}\label{as:maFull}
    For any $x,s \in S$, $q \in Q^\phi$ and any $i$ we have that $\trs[i, \infty] \in \lproRel{q}{x}{s}$ if and only if $\provRel{q}{x}{s} \in \bSeq{i}$. 
\end{assumption}

Under this assumption we still need to show that $[\pi \mapsto \trs, \pi' \mapsto \trss] \models \phi$ (which is the conclusion of $\phi^{P, \Xi}$).

\begin{strat}[t!]
	\begin{algorithmic}[1]
		\State \textbf{Input:} $\vertex \in V^*$ \textbf{where} $|\vertex| = 2i + 1$
		\If{$i = 0$}
		\State  \textbf{return} $s_0$\label{line:b1}
		\Else
		\State  $T \coloneqq \sucs{\ssSeq{i-1}}$\label{line:b2}
		\State $\hat{q} \coloneqq {\delta^\phi}^* \big[(\sTr{0}, \ssTr{0}) \cdots (\sTr{i-1},\ssTr{i-1})\big]$\label{line:b3}
		\State $C \coloneqq \{ s' \mid s' \in T \land  \provRel{\hat{q}}{\ssSeq{i-1}}{s'} \in \bSeq{i} \}$\label{line:b4}
		\If{$C \neq \emptyset$}
		\State \textbf{return} any $s' \in C$\label{line:b5}
		\Else
		\State \textbf{return} any $s' \in T$\label{line:b6}
		\EndIf
		\EndIf
	\end{algorithmic}
	
	\caption{Winning strategy for $\veri$ for the prophecy construction from \Cref{sec:constructionBuechi}.}\label{fig:strategyBuchi}
\end{strat}

\subsection{Main Invariant}

For each $i \in \nat$ we define
\begin{align*}
    \qSeqHat{i} \coloneqq  {\delta^\phi}^*\big[(\sTr{0}, \ssTr{0}) \cdots (\sTr{i-1},\ssTr{i-1})\big]
\end{align*}
as the unique automaton sequence of $\aut^\phi$ on $\trs \otimes \trss$. 
Which coincides with the state computed in the $i$th iteration in line \ref*{line:b3} of \Cref{fig:strategyBuchi}.

Similar to the safety case in \Cref{sec:appSafetyCase}, we first show the following invariant on the play $\vertex$ (under \Cref{as:maFull}).
Note that we restrict $i > 1$.

\begin{proposition}\label{prop:claimFull}
    For every $i \geq 1$ it holds that 
    \begin{align*}
    	\provRel{\qSeqHat{i}}{\ssSeq{i-1}}{\ssSeq{i}} \in \bSeq{i}.
    \end{align*}
\end{proposition}

The proof is similar to that in the safety case and proceeds by two lemmas and a subsequent induction.
It is complicated by the additional optimality constraint in the definition of $\lproRel{q}{x}{s}$.

\begin{lemma}\label{lem:indStartFull}
    There exists a state $s'_1 \in \sucs{s_0}$ such that $\trs[1, \infty] \in \lproRel{ \qSeqHat{1}}{s_0}{s'_1}$.
\end{lemma}
\begin{proof}
    Define 
    \begin{align*}
        T \coloneqq \{ t' \in \traces{\tss} \mid \trs \otimes t' \in \calL(\aut^\phi) \}.
    \end{align*}
    As assumed that $\tss \models \forall \pi.\exists \pi'. \phi$ and $\trs \in \traces{\tss}$ we get that $T \neq \emptyset$.
    For each $t \in T$, define $c(t) \coloneqq \firstvisit{F^\phi}{\aut^\phi}{\trs \otimes t}$ (for all $t \in T$ this number must be finite).
    Now pick any $t^\star \in T$ such that $c(t^\star)$ is minimal among all traces in $T$, i.e., a valid witness runs that visits an accepting state as fast as possible. 
    As $T \neq \emptyset$ there exists at least one such $t^\star$.
    Let $p^\star \in \paths{\tss}$ be such that $L(p^\star) = t^\star$.
    Now define $s_1' \coloneqq p^\star(1)$.
    As by definition $p^\star(0) = s_0$ (as this is the only initial state), we have $s'_1 \in \sucs{s_0}$ as required.

    It remains to show that $\trs[1, \infty] \in \lproRel{ \qSeqHat{1}}{s_0}{s'_1}$.
    As the witness for the existential quantification in the definition of $\lproRel{ \qSeqHat{1}}{s_0}{s'_1}$ we choose $t^\star[1, \infty]$.
    We then need to show that
    \begin{align}\label{eq:toshow1}
        \trs[1, \infty] \otimes t^\star[1, \infty] \in \calL(\aut^\phi_{ \qSeqHat{1}})
    \end{align}
    and 
    \begin{align}\label{eq:toshow2}
        \begin{split}
            &\forall s' \in \sucs{s_0}\ldot \forall t'' \in \mathit{Traces}(\tss_{s'}). \\
            &\quad\quad\trs[1, \infty] \otimes t'' \in \calL(\aut^\phi_{ \qSeqHat{1}})\\
            &\quad\quad\Rightarrow \Big(\firstvisit{F^\phi}{\aut^\phi_{ \qSeqHat{1}}}{\trs[1, \infty] \otimes t^\star[1, \infty]}  \\
            &\quad\quad\quad\quad\leq  \firstvisit{F^\phi}{\aut^\phi_{ \qSeqHat{1}}}{\trs[1, \infty] \otimes t''} \Big).
        \end{split}       
    \end{align}
    We first show (\ref{eq:toshow1}).
    By assumption we have $\trs \otimes t^\star \in \calL(\aut^\phi)$ and
    \begin{align*}
        &\trs \otimes t^\star \in \calL(\aut^\phi)\\
        &\Rightarrow (L(s_0)\trs[1,\infty]) \otimes (L(s_0)t^\star[1, \infty]) \in \calL(\aut^\phi_{q^\phi_0})\\
        &\Rightarrow (L(s_0), L(s_0)) (\trs[1, \infty] \otimes t^\star[1, \infty]) \in \calL(\aut^\phi_{q^\phi_0})\\
        &\Rightarrow  \trs[1, \infty] \otimes t^\star[1, \infty] \in \calL(\aut^\phi_{ \qSeqHat{1}}).
    \end{align*}
    Here the first implication follows as $\trs, t^\star \in \traces{\tss}$ and $s_0$ is the unique initial state, the second implication from the definition of $\otimes$ and the third from the definition of $ \qSeqHat{1}$ and the semantics of a deterministic automaton. 
    
    We now show (\ref{eq:toshow2}).
    Let $s' \in \sucs{s_0}$ and $t'' \in \traces{\tss_{s'}}$ be arbitrary such that $\trs[1, \infty] \otimes t'' \in \calL(\aut^\phi_{ \qSeqHat{1}})$.
    This implies that $\trs  \otimes L(s_0)t'' \in \calL(\aut)$ (note that $\qSeqHat{1} = \delta^\phi(q_0^\phi, (L(s_0), L(s_0)))$) and $L(s_0)t''\in \traces{\tss_{s_0}} = \traces{\tss}$.
    By construction of $T$ this implies $L(s_0)t'' \in T$ and so, by selection of $t^\star$, we get $c(t^\star) \leq c(L(s_0)t'')$.
    As we assume that $q^\phi_0 \not\in F^\phi$ we get
    \begin{align*}
    	c(t^\star) &= \firstvisit{F^\phi}{\aut^\phi}{\trs \otimes t^\star} \\
    	&= 1 + \firstvisit{F^\phi}{\aut^\phi_{ \qSeqHat{1}}}{\trs[1, \infty] \otimes t^\star[1, \infty]}.
    \end{align*}
	And similarly
    \begin{align*}
    	c(L(s_0)t'') &= \firstvisit{F^\phi}{\aut^\phi}{\trs \otimes L(s_0)t''} \\
    	&= 1 + \firstvisit{F^\phi}{\aut^\phi_{ \qSeqHat{1}}}{\trs[1, \infty] \otimes t''}.
    \end{align*}
    Combined with $c(t^\star) \leq c(L(s_0)t'')$ we get 
	\begin{align*}
		&\firstvisit{F^\phi}{\aut^\phi_{ \qSeqHat{1}}}{\trs[1, \infty] \otimes t^\star[1, \infty]}\\
		&\quad\quad \leq  \firstvisit{F^\phi}{\aut^\phi_{ \qSeqHat{1}}}{\trs[1, \infty] \otimes t''}
	\end{align*}
	as required in (\ref{eq:toshow2}).
\end{proof}

\begin{lemma}\label{lem:progressFull}
    If $t \in \lproRel{q}{x}{s}$ then there exists a state $y \in \sucs{s}$ such that $t[1, \infty] \in \lproRel{q'}{s}{y}$ where $q' = \delta^\phi(q, (t(0), L(s)))$.
\end{lemma}
\begin{proof}
    Assume $t \in \lproRel{q}{x}{s}$. 
    Let $t' \in \traces{\tss_s}$ be the witness trace in the definition of $\lproRel{q}{x}{s}$.
    In particular we get 
    \begin{align}\label{eq:cond1}
        t \otimes t' \in \calL(\aut^\phi_q).
    \end{align}
    Now define
    \begin{align*}
        T \coloneqq \{ t' \mid \exists y \in \sucs{s}\ldot &t' \in \traces{\tss_y} \;\land   \\
        &t[1,\infty] \otimes t' \in \calL(\aut^\phi_{q'}) \}.
    \end{align*}
    That is, $T$ contains all witness traces for $t[1, \infty]$ that can be generated from some successor of $s$.
    
    We first argue that $T \neq \emptyset$.
    Define $t'' \coloneqq t'[1, \infty]$. We claim $t'' \in T$.
    We thus need to show that, first $\exists y \in \sucs{s}. t'' \in \traces{\tss_y} $, and second $t[1,\infty] \otimes t'' \in \calL(\aut^\phi_{q'}$.
    The first follows easy as $t' \in \traces{\tss_s}$ and so $t'' = t'[1, \infty]$ is the trace from some successor of $s$.
    For the second, we use (\ref{eq:cond1}) and the fact that
    \begin{align*}
        t\otimes t'  \in \calL(\aut^\phi_{q})
        &\Rightarrow(t(0) t[1, \infty]) \otimes (L(s)t'')) \in \calL(\aut^\phi_{q}) \\
        &\Rightarrow (t(0), L(s)) (t[1, \infty]) \otimes t'') \in \calL(\aut^\phi_{q}) \\
        &\Rightarrow t[1, \infty] \otimes t'' \in \calL(\aut^\phi_{q'}).
    \end{align*}%
    Here the first implication holds by definition of $t''$, the second by definition $\otimes$, and the third by definition of $q'$ as $q' = \delta^\phi(q, (t(0), L(s)))$.
    
    We thus have established that $T \neq \emptyset$.
    Similar to the proof of \Cref{lem:indStartFull}, we pick an optimal trace in $T$.
    For each $t' \in T$, define $c(t') \coloneqq \firstvisit{F^\phi}{\aut^\phi_{q'}}{t[1,\infty] \otimes t'}$ (for all $t' \in T$ this number must be finite).
    Now pick any $t^\star \in T$ such that $c(t^\star)$ is minimal among all traces in $T$, i.e., a valid witness runs that visits an accepting state as fast as possible. 
    As $t^\star \in T$ we have that $t^\star \in \traces{\tss_{y^\star}}$ for some $y^\star \in \sucs{s}$.
    We claim that $y^\star$ is the witness asked for in the present lemma.
    
    It remains to show that $t[1, \infty] \in \lproRel{q'}{s}{y^\star}$.
    For the existential quantifier in the definition of $\lproRel{q'}{s}{y^\star}$ we choose $t^\star$.
    We then need to that
    \begin{align}\label{eq:toShow1}
        t[1, \infty] \otimes t^\star \in \calL(\aut^\phi_{q'})
    \end{align}
    and 
    \begin{align}\label{eq:toShow2}
        \begin{split}
            &\forall y' \in \sucs{s}. \forall t'' \in \mathit{Traces}(\tss_{y'}). \\
            &\quad\quad t[1, \infty] \otimes t'' \in \calL(\aut^\phi_{q'})\\
            &\quad\quad\Rightarrow \Big(\firstvisit{F^\phi}{\aut^\phi_{q'}}{t[1, \infty] \otimes t^\star}  \\
            &\quad\quad\quad\quad\leq  \firstvisit{F^\phi}{\aut^\phi_{q'}}{t[1, \infty] \otimes t''} \Big).
        \end{split}       
    \end{align}
    Here, (\ref{eq:toShow1}) follows directly as $t^\star \in T$.
    To show (\ref{eq:toShow2}) let $y' \in \sucs{s}$ and $t'' \in \mathit{Traces}(\tss_{y'})$ be arbitrary with $t[1, \infty] \otimes t'' \in \calL(\aut^\phi_{q'})$.
    This already gives us that $t'' \in T$.
    By construction of $t^\star$ we get that $c(t^\star) \leq c(t'')$ and so 
    \begin{align*}
        \begin{split}
            c(t^\star) = &\firstvisit{F^\phi}{\aut^\phi_{q'}}{t[1, \infty] \otimes t^\star} \\
            &\quad\quad\leq  \firstvisit{F^\phi}{\aut^\phi_{q'}}{t[1, \infty] \otimes t''} = c(t'')
        \end{split}
    \end{align*}
	as required in (\ref{eq:toShow2}).
\end{proof}

Using \Cref{lem:indStartFull} and \Cref{lem:progressFull} we are now in a position to prove \Cref{prop:claimFull}.

\begin{proof}[Proof of \Cref{prop:claimFull}]
    We prove this by induction on $i$.
    
   \textbf{The case $i = 1$:}  By \Cref{lem:indStartFull} there exists $s'_1 \in \sucs{s_0}$ such that $\trs[1, \infty] \in \lproRel{ \qSeqHat{1}}{s_0}{s'_1}$.
   	By \Cref{as:maFull} this implies that $\provRel{ \qSeqHat{1}}{s_0}{s'_1} \in \bSeq{1}$.
        In the first step of $\vertex$, $\sigma$ returns the unique initial state $s_0$, so $\ssSeq{0} = s_0$.
        In the second step and using the above we get that $s_1' \in C$ (where $C$ is computed in iteration $i = 1$ in line \ref*{line:b4} of \Cref{fig:strategyBuchi}) so $C \neq \emptyset$.
        By construction $\sigma$ will have picked some state $\ssSeq{1} \in C$ (not necessarily $s'_1$) which already implies that $\provRel{\qSeqHat{1}}{\ssSeq{0}}{\ssSeq{1}}  \in \bSeq{1}$ as required.
        
        \textbf{The case $i > 1$:} By the induction hypothesis we can assume that $\provRel{\qSeqHat{i-1}}{\ssSeq{i-2}}{\ssSeq{i-1}}  \in \bSeq{i-1}$. 
        By \Cref{as:maFull}, this implies that 
        \begin{align*}
            \trs[i-1, \infty] \in \lproRel{\qSeqHat{i-1}}{\ssSeq{i-2}}{\ssSeq{i-1}}.
        \end{align*}
        By \Cref{lem:progressFull} there exists $y \in \sucs{\ssSeq{i-1}}$ such that 
        \begin{align}
            (\trs[i-1, \infty])[1, \infty] = \trs[i, \infty] \in \lproRel{q'}{\ssSeq{i-1}}{y}
        \end{align}
        where 
        \begin{align*}
            q' = \delta^\phi(\qSeqHat{i-1}, (\trs(i-1), L(\ssSeq{i-1}))) = \qSeqHat{i}.
        \end{align*}
        As $\trs[i, \infty] \in \lproRel{\qSeqHat{i}}{\ssSeq{i-1}}{y}$ we can, again, use \Cref{as:maFull} and obtain that $\provRel{ \qSeqHat{i}}{\ssSeq{i-1}}{y} \in \bSeq{i}$.
        Now by construction of $\sigma$ this implies that $y \in C$ (where $C$ is the set computed by $\sigma$ in the $i$th iteration in line \ref*{line:b4} of \Cref{fig:strategyBuchi}) so $C \neq \emptyset$ and the strategy will have picked some successor $\ssSeq{i} \in C$ (not necessarily $y$) which already implies that $\provRel{\qSeqHat{i}}{\ssSeq{i-1}}{\ssSeq{i}} \in \bSeq{i}$ as required.
\end{proof}

\subsection{Ranking Function}

Different from the safety case, \Cref{prop:claimFull} does not directly imply that $\veri$ actually wins $\vertex$.
We need some more notation to formalize the idea that choosing a best successor (as in the definition of our prophecies) actually guarantees a repeated visit to $F^\phi$.
We employ an, at first glance, unrelated function that serves as a ranking argument for the next visit to an accepting state, i.e., in each step this variant (into a well-ordered set $\nat$) decreases, thereby guaranteeing infinity many visits to $F^\phi$.

For $q \in Q^\phi, s \in S$ and $t \in \Sigma^\omega$ we define
\begin{align*}
    \reponse{q}{x}{t} \coloneqq \{ t'  \mid \exists &s' \in \sucs{x}. t' \in \traces{\tss_{s'}} \; \land \\
    &t \otimes t' \in \calL(\aut^\phi_q) \}.
\end{align*}
I.e., as all traces starting in some successor of $x$ that can serve as a witness for $t$.
We then define
\begin{align*}
    \optPath{q}{x}{t} \coloneqq \min\limits_{t' \in \reponse{q}{x}{t}} \firstvisit{F^\phi}{\aut^\phi_q}{t \otimes t'}
\end{align*}
as the best possible response to trace $t$ from $q, x$.
If $\reponse{q}{x}{t} = \emptyset$ the minimum ranges over an empty set so $\optPath{q}{x}{t}  = \infty$.
We can establish the following relation between the fact that $t \in \lproRel{q}{x}{s}$ and the definition of $\optPath{q}{x}{t}$.

\begin{proposition}\label{prop:progress}
    For all $q \in Q^\phi$, $s \in \sucs{x}$, and $t \in \Sigma^\omega$ such that $t \in \lproRel{q}{x}{s}$ the following hold:
    \begin{enumerate}
        \item $\optPath{q}{x}{t} < \infty$, and
        \item If $q \not\in F^\phi$ then 
        \begin{align*}
            \optPath{q'}{s}{t[1,\infty]} < \optPath{q}{x}{t}
        \end{align*}
        where $q' = \delta^\phi(q, (t(0), L(s)))$.
    \end{enumerate}
\end{proposition}
\begin{proof}
    By definition of $\lproRel{q}{x}{s}$ and as $t \in \lproRel{q}{x}{s}$, we get a witness trace $t^\dagger \in \traces{s}$ such that
    \begin{align}\label{eq:Cond1}
        t \otimes t^\dagger \in \calL(\aut^\phi_q)
    \end{align}
    and 
    \begin{align}\label{eq:Cond2}
        \begin{split}
            &\forall s' \in \sucs{x}. \forall t'' \in \mathit{Traces}(\tss_{s'}). t \otimes t'' \in \calL(\aut^\phi_q)\\
            &\quad\quad\Rightarrow \Big(\firstvisit{F^\phi}{\aut^\phi_q}{t \otimes t^\dagger}  \\
            &\quad\quad\quad\quad\leq  \firstvisit{F^\phi}{\aut^\phi_q}{t \otimes t''} \Big).
        \end{split}       
    \end{align}
    We now show three separate equations:
    
    \textbf{First Eq:}
    We show that 
    \begin{align}\label{eq:eq}
        \optPath{q}{x}{t} = \firstvisit{F^\phi}{\aut^\phi_q}{t \otimes t^\dagger}.
    \end{align}
    That is, $t^\dagger$ is one of the optimal traces in the definition of $\optPath{q}{x}{t}$.
    We show this equality by showing the $\leq$-direction and $\geq$-direction separately.
    For the $\leq$-direction, we use (\ref{eq:Cond1}), the fact that $t^\dagger \in \traces{\tss_s}$, and $s \in \sucs{x}$, to deduce that $t^\dagger \in \reponse{q}{x}{t}$.
    As $\optPath{q}{x}{t}$ is defined as the minimum we thus have $\optPath{q}{x}{t} \leq \firstvisit{F^\phi}{\aut^\phi_q}{t \otimes t^\dagger}$.
    In particular $\optPath{q}{x}{t}  < \infty$ so we have already proven the first point of this proposition.
    For the $\geq$-direction, take an arbitrary $t' \in \reponse{q}{x}{t}$. 
    Now by definition of $\reponse{q}{x}{t}$ we get that $t' \in \traces{\tss_{s'}}$ for some $s' \in \sucs{x}$ and $t \otimes t' \in \calL(\aut^\phi_q)$.
    If we plug $s'$ and $t'$ into (\ref{eq:Cond2}), we get 
    \begin{align*}
        \begin{split}
            \firstvisit{F^\phi}{\aut^\phi_q}{t \otimes t^\dagger}  \leq  \firstvisit{F^\phi}{\aut^\phi_q}{t \otimes t'}.
        \end{split}       
    \end{align*}
    As this holds for every $t' \in  \reponse{q}{x}{t}$, it also holds for the minimum over $\reponse{q}{x}{t}$ so $\optPath{q}{x}{t} \geq \firstvisit{F^\phi}{\aut^\phi_q}{t \otimes t^\dagger}$.
    
    \textbf{Second Eq:}
    Define $t^\star \coloneqq t^\dagger[1, \infty]$ (so $t^\dagger = L(s)t^\star$).
    We show that 
    \begin{align}\label{eq:third}
    	\optPath{q'}{s}{t[1, \infty]} \leq \firstvisit{F^\phi}{\aut^\phi_{q'}}{t[1,\infty] \otimes t^\star}.
    \end{align}
	By definition of $\mathit{opt}$ (which is defined as the minimum) it suffices to show that 
    \begin{align*}
        t^\star \in \reponse{q'}{s}{t[1,\infty]}.
    \end{align*}
    To show this we need to show that $\exists s' \in \sucs{s}\ldot t^\star \in \traces{\tss_{s'}}$ and $t[1, \infty] \otimes t^\star \in \calL(\aut^\phi_{q'})$.
    The first follows directly, as $t^\dagger \in \traces{\tss_s}$ we must have that $t^\star = t^\dagger[1, \infty] \in \traces{\tss_{s'}}$ for some $s' \in \sucs{s}$.
    For the second we can deduce:
    \begin{align*}
        t\otimes t^\dagger  \in \calL(\aut^\phi_{q})
        &\Rightarrow(t(0) t[1, \infty]) \otimes (L(s)t^\star)) \in \calL(\aut^\phi_{q}) \\
        &\Rightarrow (t(0), L(s)) (t[1, \infty]) \otimes t^\star) \in \calL(\aut^\phi_{q}) \\
        &\Rightarrow t[1, \infty] \otimes t^\star \in \calL(\aut^\phi_{q'}).
    \end{align*}%
    Here $t\otimes t^\dagger  \in \calL(\aut^\phi_{q})$ holds by (\ref{eq:Cond1}), the first implication by definition of $t^\star$, the second by definition $\otimes$, and the third by the automaton semantics and as $q' = \delta^\phi(q, (t(0), L(s)))$.
    
    \textbf{Third Eq:}
    As a third step, we derive that
    \begin{align}\label{eq:plus1}
        \begin{split}
            &\firstvisit{F^\phi}{\aut^\phi_q}{t \otimes t^\dagger} \\
            &\quad= \firstvisit{F^\phi}{\aut^\phi_q}{t(0)t[1,\infty] \otimes L(s)t^\star}\\
            &\quad= 1 +  \firstvisit{F^\phi}{\aut^\phi_{q'}}{t[1,\infty] \otimes t^\star}.
        \end{split}
    \end{align}
    Here the first equality follows from the definition of $t^\star$ and the second from the definition of $\mathit{firstVisit}$, the fact that $q' = \delta^\phi(q, (t(0), L(s)))$ and crucially the assumption that $q \not\in F^\phi$.
    
    \textbf{Combination:}
    We can now put all the pieces together and derive
    \begin{align*}
        \optPath{q}{x}{t} &= \firstvisit{F^\phi}{\aut^\phi_q}{t \otimes t^\dagger}\\
        &= 1 +  \firstvisit{F^\phi}{\aut^\phi_{q'}}{t[1,\infty] \otimes t^\star}\\
        &> \firstvisit{F^\phi}{\aut^\phi_{q'}}{t[1,\infty] \otimes t^\star}\\
        &\geq  \optPath{q'}{s}{t[1, \infty]}.
    \end{align*}
    The first equality follows from \ref{eq:eq}, the second one from \ref{eq:plus1}.
    The strict inequality is by simple arithmetic and the last inequality follows from \ref{eq:third}.
    This concludes the proof.
\end{proof}

Using \Cref{prop:claimFull} and \Cref{prop:progress} we can show that $\veri$ wins $\vertex$:

\begin{proposition}\label{prop:resFull}
    Play $\vertex$ is won by $\veri$.
\end{proposition}
\begin{proof}
    We need to show that $\big[\pi \mapsto \trs, \pi' \mapsto \trss\big] \models \phi$.
    That is, the unique run of $\aut^\phi$ on $\trs \otimes \trss$, which is $\qSeqHat{0}\qSeqHat{1}\cdots$, visits states in $F^\phi$ infinity any times. 
    For any $i \geq 1$ define $r_i \in \nat$ by 
    \begin{align*}
        r_i \coloneqq \optPath{\qSeqHat{i}}{\ssSeq{i-1}}{\trs[i,\infty]}.
    \end{align*}
    We claim that for every $i \geq 1$, $r_i < \infty$.
    Moreover for every $i$, either $\qSeqHat{i} \in F^\phi$ or $r_i > r_{i+1}$. 
    This would already imply that $\qSeqHat{i} \in F^\phi$ for infinity many $i$ as $r_i$ can only decrease a finite number of times between two steps where $\qSeqHat{i} \in F^\phi$.
    
    Fix any $i \geq 1$. By \Cref{prop:claimFull} we have that 
    \begin{align*}
        \provRel{\qSeqHat{i}}{\ssSeq{i-1}}{\ssSeq{i}} \in \bSeq{i}.
    \end{align*}
    Using \Cref{as:maFull} this implies that 
    \begin{align*}
        \trs[i,\infty] \in \lproRel{\qSeqHat{i}}{\ssSeq{i-1}}{\ssSeq{i}}.
    \end{align*}
    By the first part in \Cref{prop:progress} and the definition of $r_i$ this already gives us $r_i < \infty$.
    To show the second part of the claim assume that $\qSeqHat{i} \not\in F^\phi$.
    Using the second part in \Cref{prop:progress} we thus get 
    \begin{align*}
        &r_{i+1}= \optPath{\qSeqHat{i+1}}{\ssSeq{i}}{\trs[i+1,\infty]} \\
        &\quad\quad< \optPath{\qSeqHat{i}}{\ssSeq{i-1}}{\trs[i,\infty]} = r_{i}
    \end{align*}
    as required.
\end{proof}

\Cref{prop:resFull} concludes the proof of \Cref{theo:compgeneral}.

\section{Prophecy Construction Beyond Deterministic Büchi}\label{app:constructionBeyond}

In this section, we formalize the prophecy construction beyond deterministic Büchi.
We only sketch an idea of the proof as the main ideas can already be found in the proof of \Cref{theo:compgeneral} for deterministic Büchi automata. 
We first consider the case of a single-pair Rabin automaton where states in $F^\phi$ should be visited infinitely many times and states in $B^\phi$ at most finitely many times.
Similar to the definition of $\firstvisit{F^\phi}{\aut^\phi}{t}$ we define $\lastvisit{B^\phi}{\aut^\phi}{t} \in \nat \cup \{\infty, \bot\}$ as the last time that a state in  $B^\phi$ has been visited.
We define it to be $\infty$ if $B^\phi$ is visited infinity many times (so there is no last visit) and $\bot$ if it \emph{never} visits a state in $B^\phi$. (Note that if $\lastvisit{B^\phi}{\aut^\phi}{t} = 0$ a visit fo $B^\phi$ occurs in the initial state).

With those definitions at hand we can define two sets of traces.
We define $\lproRel{q}{x}{s}^1$ as\\
\scalebox{0.92}{
\parbox{\columnwidth}{
     \begin{align*}
        \Big\{ t &\in \Sigma^\omega \mid \exists t' \in \mathit{Traces}(\tss_{s}). \\
        &t \otimes t' \in \calL(\aut^\phi_q) \land \lastvisit{B^\phi}{\aut^\phi_q}{t \otimes t'} = \bot \,\land \\
        &\Big[\forall s' \in \sucs{x}\ldot \forall t'' \in \mathit{Traces}(\tss_{s'}). \\
        &\quad \big(t \otimes t'' \in \calL(\aut^\phi_q) \land \lastvisit{B^\phi}{\aut^\phi_q}{t \otimes t''} = \bot\big) \\
        &\quad\Rightarrow \firstvisit{F^\phi}{\aut^\phi_q}{t \otimes t'}  \leq  \firstvisit{F^\phi}{\aut^\phi_q}{t \otimes t''}\Big] \Big\}.
    \end{align*}
}
}\\
This definition corresponds to the first point in our informal definition. 
A trace is in $\lproRel{q}{x}{s}^1$ if there exists a witness trace that never visits a state in $B^\phi$ (which we can express as $\lastvisit{B^\phi}{\aut^\phi_q}{t \otimes t'} = \bot$) such that every alternative successor $s'$ and trace from $s'$ that also never visits $B^\phi$ takes at least as long to visit a state in $F^\phi$ for the first time.

We define $\lproRel{q}{x}{s}^2$ as\\
\scalebox{0.92}{
    \parbox{\columnwidth}{
        \begin{align*}
            \Big\{ t &\in \Sigma^\omega \mid \Big[\forall s' \in \sucs{x}\ldot \forall t'' \in \mathit{Traces}(\tss_{s'})\ldot \\
            &\quad \big(t \otimes t'' \in \calL(\aut^\phi_q) \Rightarrow  \lastvisit{B^\phi}{\aut^\phi_q}{t \otimes t''} \neq \bot\big) \Big] \land  \\
            &\quad\exists t' \in \mathit{Traces}(\tss_{s})\ldot t \otimes t' \in \calL(\aut^\phi_q) \, \land \\
            &\quad\Big[\forall s' \in \sucs{x}\ldot \forall t'' \in \mathit{Traces}(\tss_{s'})\ldot  t \otimes t'' \in \calL(\aut^\phi_q) \\
            &\quad\quad\Rightarrow \lastvisit{B^\phi}{\aut^\phi_q}{t \otimes t'}  \leq  \lastvisit{B^\phi}{\aut^\phi_q}{t \otimes t''}\Big] \Big\}.
        \end{align*}
    }
}\\
The definition states that from no successor is it possible to construct a witness such that the run of $\aut^\phi$ never visits a state in $B^\phi$.
Additionally there should be a witness trace from $s$ such that for every possible successor $s'$ of $x$ and all traces $t''$ from $s'$ the last visit to $B^\phi$ (which must be finite as both $t'$ and $t''$ are winning) occurs at least as fast on $t \otimes t'$ as on $t \otimes t''$.

Finally we can define the prophecies for one-pair Rabin automata as the union of both languages, i.e., 
\begin{align*}
    \lproRel{q}{x}{s} \coloneqq \lproRel{q}{x}{s}^1 \cup \lproRel{q}{x}{s}^2.
\end{align*}
We again need to argue that the resulting languages are $\omega$-regular.
Here we can employ the same idea as in \Cref{prop:omegaReagular} to construct QPTL representations of $ \lproRel{q}{x}{s}^1$ and  $\lproRel{q}{x}{s}^2$ and use the fact that $\omega$-regular languages are closed under union.

\section{Reducing the Number Of Prophecies to $\calO(|S| \log |S|)$}\label{app:nlogn}

In this section we sketch a proof to reduce the number of prophecies needed to obtain a proof of the following:

\balance
\propnlogn*
\begin{proof}
	In the construction from \Cref{sec:compGeneral} the number of prophecies scales quadratically in the number of states in $\tss$, as, for each state $x$, we identify all successor states $s$ that are optimal. 
	That is, for a given trace $t$, state $x$ and automaton state $q$, the set $\{s \mid t \in \lproRel{q}{x}{s}\}$ are all states that are ``optimal''.
	For an actual strategy, it is, however, not necessary to obtain the \emph{set} of optimal successors but only a fixed one of those optimal successors. 
	This allows for exponentially fewer predicates (compared to the construction in \Cref{sec:compGeneral}).
	For simplicity, we stick with the case where $\phi$ is recognizable by a deterministic Büchi automaton.
	In our previous construction we defined a prophecy $\lproRel{q}{x}{s}$ only for states $x, s$ where $s \in \sucs{x}$. 
	For simplify we set $\lproRel{q}{x}{s} = \emptyset$ if $s \not\in \sucs{x}$. 
	This way, $\lproRel{q}{x}{s}$ is defined for all $q \in Q^\phi$ and $x, s \in S$.
	The set $\{\lproRel{x}{q}{s}  \}_{q \in Q^\phi, x, s \in S}$ is still complete.

	Assume that $S = \{s_1, \ldots, s_n\}$ and assume w.l.o.g., that $n = 2^m$ for some some $m$ (we can always add unreachable states to the system).
	As a first step, we define modified prophecy variables that only hold if the successor is optimal \emph{and} also minimal (where we order states based on their index).
	For a set of traces $A$, we write $\overline{A}$ for the complement of $A$.
	For $q \in Q^\phi, x,  \in S, s_i \in S$ we define
	\begin{align*}
		\lpromin{x}{q}{s_i} \coloneqq \lproRel{q}{x}{s_i} \cap \bigcap\limits_{j < i} \overline{\lproRel{q}{x}{s_j}}.
	\end{align*}%
	That is, a trace $t$ is in $\lpromin{x}{q}{s}$ if $s$ is a successor for the $\exists$-player that is optimal (in the sense that $t \in \lproRel{q}{x}{s}$) but no smaller state is optimal. 
	It is easy to see that $\{  \lpromin{x}{q}{s}  \}_{q \in Q^\phi, x, s \in S}$ (suing the same proof that showed that $\{  \lproRel{x}{q}{s}  \}_{q \in Q^\phi, x, s \in S}$ is complete , cf.~\Cref{theo:compgeneral}).
	
	We further note that $\{\lpromin{x}{q}{s}\}_{s \in S}$ forms a \emph{partition} of $\bigcup_{s \in S} \lproRel{q}{x}{s}$. 
	In particular, at most one of the prophecies can hold at any given time, i.e., for every trace $t$ there exists \emph{at most once} $s$ with $t \in \lpromin{x}{q}{s}$.
	We now encode this unique state $s$ in binary.
	Let $\encode{\cdot} :  \{1, \ldots, n\} \to \mathbb{B}^m$ and $\decode{\cdot} : \mathbb{B}^m \to \{1, \ldots, n\}$ be some encoding of the states, i.e., for every $i \in \{1, \ldots, n\}$, $\decode{\encode{i}} = i$.
	We write $\encode{i}_j \in \mathbb{B}$ for the $j$th bit in $\encode{i}$.
	For $q \in Q^\phi$, $x \in S$ and $j \in \{1, \ldots, m\}$ we now define
	\begin{align*}
		\lpromod{x}{q}{j} \coloneqq \bigcup_{\substack{i \in \{1, \ldots, n\}\\ \encode{i}_j = \top}} \lpromin{x}{q}{s_i}.
	\end{align*}%
	That is $\lpromod{x}{q}{j}$ contains all traces that are in $\lpromin{x}{q}{s_i}$ for state $s_i$ where the $j$th position in the binary encoding of $i$ is set to true. 
	
	In particular, for every $q \in Q^\phi, x \in S$ we can recover the sets $\lpromin{x}{q}{s_i}$ for each $s_i \in S$ as a boolean combination of $\lpromod{x}{q}{1}, \ldots, \lpromod{x}{q}{m}$.
	We write $\mathit{ite}(c, A, B)$ for $A$ if $b = \top$ and $B$ otherwise.
	With this notation we get that 
	\begin{align}\label{eq:binary}
		\lpromin{x}{q}{s_i} = \bigcap_{j = 1}^m \mathit{ite}(\encode{i}_j,  \lpromod{x}{q}{j} , \overline{ \lpromod{x}{q}{j} } ).
	\end{align}
	Note that here it is crucial that for every $q \in Q^\phi, x \in S$ the sets $\{\lpromin{x}{q}{s}\}_{s \in S}$ are pairwise disjoint.
	
	We claim that $\{  \lpromod{x}{q}{j}  \}_{q \in Q^\phi, x \in S, j \in \{1, \ldots, m\}}$ is a complete set of prophecies.
	Note that this set has the desired size of $\calO(|S| \log |S|)$.
	The idea to construct a strategy for the $\exists$-player is similar to before but now uses the binary encoding.
	In each step (where $x$ is the current system state and $q$ the current state of $\aut^\phi$), the strategy checks the truth value of $\lpromod{x}{q}{1}, \ldots, \lpromod{x}{q}{m}$ via the prophecy variable (which we can assume to be set correctly).
	Let $b \in \mathbb{B}^m$ be this vector.
	The strategy then selects $s_{\decode{b}}$ as its successor state. 
	As established in (\ref{eq:binary}) this decoding recovers the sets $\lpromin{x}{q}{s_i}$, so the strategy selects the successor that is optimal (as identified in $\lpromin{x}{q}{s_i}$).
\end{proof}

\fi

\end{document}